\documentclass[runningheads]{llncs}

\usepackage{subcaption}
\usepackage[T1]{fontenc}
\usepackage{graphicx}
\usepackage{amssymb,amsmath}
\usepackage{xcolor}
\usepackage{mathtools}
\usepackage{booktabs}
\usepackage{array,ragged2e}
\usepackage[skins,breakable]{tcolorbox}
\newtcolorbox{mybox}[2][]{%
    toprule=0pt,
    bottomrule=0pt,
    rightrule=0pt,
    leftrule=0pt,
    #1
}
\usepackage{thmtools}

\usepackage{tikz}
\usetikzlibrary{positioning,automata}
\usepackage{wrapfig}
\allowdisplaybreaks

\spnewtheorem{assumption}{Assumption}{\bfseries}{\itshape}

\renewcommand{\SS}[1]{{\color{cyan}\textbf{SS:}#1}}

\newenvironment{claimproof}[1]{\par\noindent\textit{Proof of claim:}\space#1}{\hfill $\blacksquare$}
\newcommand{\set}[1]{\left\lbrace #1\right\rbrace}
\newcommand{\modelsAS}{\models_{\mathsf{a.s.\!}}}

\renewcommand{\Pr}{\mathrm{Pr}}
\newcommand{\supp}{\mathsf{Supp}}
\newcommand{\pref}{\mathit{Pref}}

\newcommand{\zug}[1]{\langle #1\rangle}
\newcommand{\stam}[1]{}

\newcommand{\G}{\mathcal{G}}
\newcommand{\V}{V}
\newcommand{\E}{E}
\newcommand{\pathsfin}{\mathit{Paths}_{\mathsf{fin}}}
\newcommand{\pathsinf}{\mathit{Paths}_{\mathsf{inf}}}
\newcommand{\pol}{\pi}
\newcommand{\id}{\mathsf{id}}

\newcommand{\Pol}{\Pi}
\newcommand{\SAP}{\widetilde{\Pi}}
\newcommand{\PolFH}[1]{\Pi^{\mathrm{FH}}_N}
\newcommand{\PolS}[1]{\Pi^{\mathrm{S}}}
\newcommand{\PolP}{\Pi^{\mathrm{Path}}}

\newcommand{\SAPolFH}[1]{\widetilde{\Pi}^{\mathrm{FH}}_N}
\newcommand{\SAPolS}[1]{\widetilde{\Pi}^{\mathrm{S}}}
\newcommand{\SAPolP}{\widetilde{\Pi}^{\mathrm{Path}}}

\newcommand{\path}{\rho}
\newcommand{\vinit}{v_{\mathsf{init}}}
\newcommand{\Cyl}{\mathrm{Cyl}}

\newcommand{\spec}{\varphi}

\newcommand{\reach}{\mathit{Reach}}
\newcommand{\safe}{\mathit{Safe}}

\newcommand{\always}{\mathrm{G}}
\newcommand{\eventually}{\mathrm{F}}

\newcommand{\Distribution}{\mathcal{D}}

\newcommand{\implementation}{\mathcal{I}}
\newcommand{\specsafe}{\spec^\mathsf{safe}}
\newcommand{\speclive}{\spec^\mathsf{live}}
\newcommand{\Conv}{\text{Conv}}

\usepackage{orcidlink}
\newcommand{\sched}{\sigma}

\newcommand{\sh}{\mathsf{shld}}

\newcommand{\thetatwo}[2]{\theta_{#1, #2}}

\renewcommand{\L}{\mathcal{L}}

\makeatletter
\renewcommand\p@subfigure{}
\makeatother
\usepackage[capitalise]{cleveref}
\begin{document}
%
%\title{Distributed Control for Multi-Parity Objectives}
\title{Decoupled Planning for\\ Multiple Omega-Regular Objectives\thanks{
This work is funded by the following grants:
European Research Council under Grant No.:
ERC-2020-AdG 101020093, 
ISF grant no. 1679/21,
grant RYC2024-049116, 
MICIU/AEI/10.13039/501100011033, the ESF+, and Volkswagen Foundation within its Momentum framework under
project no. 9C283.\\
The last author was employed at Institute of Science and Technology Austria for part of this work.}}
\author{
Guy Avni\inst{1}\orcidlink{0000-0001-5588-8287}\and
Thomas A. Henzinger\inst{2}\orcidlink{0000-0001-6077-7514} \and
Kaushik Mallik\inst{3} \orcidlink{0000-0001-9864-7475}\and
Suman Sadhukhan\inst{4}\orcidlink{0000-0002-4802-6803} \and
K. S. Thejaswini\inst{5}\orcidlink{0000-0001-6077-7514}}
 \institute{Department of Computer Science, University of Haifa, Israel\\
 \email{gavni@cs.haifa.ac.il}\and
 Institute of Science and Technology Austria, Klosterneuberg, Austria
\\
 \email{tah@ista.ac.at}\and 
 IMDEA Software Institute, Spain
 \\
 \email{kaushik.mallik@imdea.org}\and
Clausthal University of Technology, Germany
 \\
 \email{suman.sadhukhan@tu-clausthal.de}\and
Université Libre de Bruxelles, Belgium
\\
\email{thejaswini.raghavan@ulb.be}}

\authorrunning{G. Avni, T. A. Henzinger, K. Mallik, S. Sadhukhan, and K. S. Thejaswini}
%% First names are abbreviated in the running head.
%% If there are more than two authors, 'et al.' is used.
%%
%\institute{Princeton University, Princeton NJ 08544, USA \and
%Springer Heidelberg, Tiergartenstr. 17, 69121 Heidelberg, Germany
\maketitle              % typeset the header of the contribution
\begin{abstract}
We study the problem of generating paths on a graph that satisfy a collection of $\omega$-regular objectives. We propose a decoupled framework in which each objective is assigned to an independent agent that selects a \emph{local} policy, while a scheduler---oblivious to the graph and objective---dynamically composes these policies into a single path. We ask when such a composition satisfies all objectives, assuming their conjunction is realizable.
The framework enables modular policy design but raises fundamental compositional challenges. We show that even extremely fair deterministic schedulers do not ensure correctness, and that \emph{stochastic} schedulers, while necessary, are insufficient without coordination. For safety objectives, we demonstrate that fully decentralized implementations are impossible, and we introduce a protocol for synchronizing on maximal safe actions.
For non-safety objectives, we introduce \emph{conventions}---simple, a priori restrictions agreed upon before the graph or objectives are revealed---that guarantee satisfaction of all objectives when followed by all agents. We characterize minimally restrictive conventions for major subclasses of $\omega$-regular objectives. In particular, B\"uchi objectives admit universal composition of finite-memory policies without scheduler communication; co-B\"uchi objectives require only knowledge of whether the agent was scheduled; and parity objectives additionally require knowledge of which agent was scheduled.

%\todo{To recycle the following:}
%Advantages of our framework include a ``separation of concerns'' when designing policies as well as modularity:  each local policy can be replaced or extended without complete recomputation. 
%We show that universal composition is not possible, and we identify required structural restrictions on the objectives and the information available to the local policies. Interestingly, for B\"uchi objectives, we prove that policies requiring no communication with the scheduler can be composed universally as long as they are finite-memory policies. For co-B\"uchi objectives, policies must only know whether they were scheduled at each step, and for parity objectives, the policies additionally need to know \emph{which} agent was scheduled in order to adapt their behaviour.

\end{abstract}
\newpage
\section{Introduction}
\label{sec:intro}

The graph traversal problem takes as input a finite directed graph and a temporal specification over sequences of vertices, and the goal is to compute a traversal \emph{policy} for selecting the next vertex based on the history of the past vertices, such that the resulting infinite path fulfills the given specification.
A primary application of graph traversal is robotic path planning~\cite{guruji2016time}, and other applications include parsing natural languages using stochastic grammars~\cite{klein2003parsing}, and informational search with online learning~\cite{kagan2014group}. 
Often, the specification requires simultaneous fulfillment of multiple, competing objectives; for example, a robot may need to simultaneously patrol a given location \textit{and} visit the charging station intermittently. The traditional approach to solve such problems is to compute a single, monolithic policy for the conjunction of all objectives. 

In this paper, we introduce a {\em decoupled} framework for graph traversals. 
Consider a collection of objectives $\varphi_1,\ldots, \varphi_N$ that a path needs to satisfy. 
We describe the framework as an interaction among $N$ agents. For $1 \leq i \leq N$, Agent~$i$ is only responsible for satisfying objective $\varphi_i$ and does so by selecting a {\em policy} that achieves $\varphi_i$. 
At runtime, a fixed {\em scheduler} composes the policies by scheduling one of the agents to choose the next vertex at each step. 
Our goal is to ensure that the path generated by the composition satisfies all objectives (either deterministically or {\em almost surely}) while achieving maximal decoupling by keeping the agents' choices and scheduler as independent as possible. Before we make these goals more concrete and discuss the features that they lead to, we first illustrate the framework and its challenges. % in the following example.
\begin{figure}[ht]
\centering
\begin{subfigure}{.45\linewidth}
  \centering
\begin{tikzpicture}[>=stealth,thick,scale=0.8,
    every state/.style={circle,draw,minimum size=8mm,inner sep=1pt,font=\small}]
  % Nodes
  \node[state] (l) at (0,0) {$\ell$};
  \node[state] (t1) at (1.2,1.2) {$t_1$};
  \node[state,initial right] (t2) at (3.2,1.2) {$t_2$};
  \node[state] (r) at (4.4,0) {$r$};
  \node[state] (b2) at (3.2,-1.2) {$b_2$};
  \node[state] (b1) at (1.2,-1.2) {$b_1$};

  % Black bidirectional edges
  \draw[<->,black] (l) -- (t1);
  \draw[<->,black] (t1) -- (t2);
  \draw[<->,black] (t2) -- (r);
  \draw[<->,black] (r) -- (b2);
  \draw[<->,black] (b2) -- (b1);
  \draw[<->,black] (b1) -- (l);
  \draw[->,black,looseness=5,out=155,in=210] (l) to (l);
  \draw[->,black,looseness=5,out=25,in=-25] (r) to (r);

  \draw[->,red,dashed,thick,bend left=25] (t1) to (l);
  \draw[->,red,dashed,thick,bend left=25] (t2) to (t1);
  \draw[->,red,dashed,thick,bend left=25] (r) to (t2);
  \draw[->,red,thick,dashed,bend left=-25] (b1) to (l);
  \draw[->,red,thick,dashed,bend left=-25] (b2) to (b1);
  \draw[->,red,thick,dashed,looseness=6,out=140,in=220] (l) to (l);
  % \draw[->,red,thick,bend left=25] (r) to[out=200,in=340,looseness=6] (r); % self-loop red
  % \draw[->,red,thick,bend left=25] (l) to[out=140,in=220,looseness=60] (l); % self-loop red
  %\draw[->,red,thick] (l) edge[loop above] node {} (l);

  \draw[->,blue,dotted,thick,bend right=-25] (t1) to (t2);
  \draw[->,blue,dotted,thick,bend right=-25] (l) to (t1);
  \draw[->,blue,thick,dotted,bend right=-25] (t2) to (r);
  \draw[->,blue,thick,dotted,bend right=25] (b2) to (r);
  \draw[->,blue,thick,dotted,bend right=25] (b1) to (b2);
  \draw[->,blue,thick,dotted,looseness=6,out=35,in=-35] (r) to (r);
%  \draw[->,blue,thick,bend right=25] (r) to[out=20,in=-20,looseness=6] (r); % self-loop blue
 % \draw[->,blue,thick,bend right=-25] (l) to[out=160,in=-160,looseness=6] (l); % self-loop blue
\end{tikzpicture}
\caption{Reachability objectives given by $T_1 =\set{\ell}$ and $T_2 = \set{r}$, and policies $\pi_1$  dashed {\color{red} red} and $\pi_2$ dotted {\color{blue} blue}.}\label{subfig:Leftintro}
\end{subfigure}
\hspace{0.5cm}
\begin{subfigure}{.45\linewidth}
  \centering
\begin{tikzpicture}[>=stealth,thick,node distance=1.2cm]

%\tikzstyle{state}=[circle,draw,minimum size=7mm,inner sep=0pt]
\tikzstyle{target1}=[circle,draw,minimum size=7mm,inner sep=0pt,fill=red!30]
\tikzstyle{target2}=[circle,draw,minimum size=7mm,inner sep=0pt,fill=blue!30]

\stam{
% Nodes
\node[state] (v) {$v$};
\node[state,above right of=v] (a1) {$a_1$};
\node[state,below right of=v] (a2) {$a_2$};
\node[target1,right of=a1] (b1) {$b_1$};
\node[target2,right of=a2] (b2) {$b_2$};
}
\node[state,initial above] (v) {$v$};
\node[state, right of=v] (a1) {$a_1$};
\node[target1,right of=a1] (b1) {$b_1$};
\node[state,left of=v] (a2) {$a_2$};
\node[target2,left of=a2] (b2) {$b_2$};

% Edges
\draw[->] (v) edge[bend right=15] (a1);
\draw[->] (v) edge[bend right=15] (a2);

\draw[->] (a1) edge[loop above] ()
          (a1) -- (b1)
          (a1) edge[bend right=15] (v);

\draw[->] (a2) edge[loop above] ()
          (a2) -- (b2)
          (a2) edge[bend right=15] (v);

\draw[->] (b1) edge[bend left=30] (v);
\draw[->] (b2) edge[bend right=30] (v);

\end{tikzpicture}
\caption{A graph with two B\"uchi objectives given by accepting vertices $T_1 = \set{b_1}$ and $T_2 = \set{b_2}$.}\label{subfig:Rightintro}
\end{subfigure}
%\caption{}
%\label{fig:intro}
\end{figure}
\begin{example}
\label{ex:deterministic versus stochastic scheduler}
Consider a coffee-serving robot in an office modeled as in ~\ref{subfig:Leftintro}.
The objective $\varphi_1$ requires ``reach and serve coffee at $\ell$'' and the objective $\varphi_2$ requires ``reach $r$.'' 
Suppose Agent~$1$ and Agent~$2$ choose policies $\pol_1$ and $\pol_2$ that always take the shortest paths towards their respective goal states.
Clearly, if the robot follows only $\pol_1$ or only $\pol_2$, respectively, $\varphi_1$ and $\varphi_2$ are satisfied. 
We describe two schedulers for composing the policies at runtime. 

First, consider a \emph{deterministic} scheduler $\sched_{\mathsf{d}}$ that alternates between $\pol_1$ and $\pol_2$: at odd time steps, the robot follows $\pol_1$, and at even time steps, it follows $\pol_2$. 
The composition does not satisfy either objective. 
For example, the path starting from $t_2$ forever alternates between $t_1$ and $t_2$.
Notably, $\sched_{\mathsf{d}}$ fulfills the \emph{extreme fairness} property~\cite{pnueli1983extremely}, because each agent is given control infinitely often. 
%This example shows that fairness is not enough when it comes to our decoupled framework.
The example shows that in order to guarantee that all objectives are satisfied, our framework requires schedulers with an even stronger form of fairness.

Now consider a \emph{stochastic} scheduler $\sched_{\mathsf{s}}$ that, at each point in time, chooses which policy to follow uniformly at random. 
This scheduler fulfills a stronger notion of fairness, called \emph{probabilistic fairness}~\cite{de1998formal} (also related to \emph{$\alpha$-fairness}~\cite{pnueli1993probabilistic}).
The composition of $\pi_1$ and $\pi_2$ using $\sched_{\mathsf{s}}$ now generates a random path that almost surely visits each of $\ell, t_1, t_2$, and $r$ infinitely often, thus both objectives are satisfied almost surely. 
\qed
\end{example}

We postpone the justification of our design choices of schedulers and policies until the end of the introduction.

%\paragraph{Features and design choices.} 
The advantages of our new framework stem from its modularity; it enables a separation of concerns in designing policies and allows changing, deleting, or adding any number of 
%policy modules %Guy: to save the line break.
policies
without any extra computational overhead, offering a lightweight, plug-and-play composition operation. 
We point out specific advantages over the monolithic approach:
(i)~{\em parallel design}: each policy is designed independently, which enables, for example, designed on separate CPUs to increase scalability, or by different vendors increasing flexible design; 
(ii)~{\em robustness}: when an objective changes, only the relevant module needs to change and the other modules can stay fixed; 
(iii)~{\em iterative design}: a realistic design procedure often requires addition of objectives to the original objectives, which can be simply added as a new module; 
(iv)~{\em dynamic changes}: our framework allows deleting or adding objectives to a system that has already been deployed (a ``patch'').

% -------- OLD ----------
%\begin{center}
%{\em Our goal is to design, for each family of objectives, decoupling schemes that maximize independence and minimize offline and runtime coordination.}
%\end{center}
% ---------------------

%Indeed, independence and minimal coordination lead to robust components, namely components that maximize the changes to other components in the system that they can withstand without change. We discuss our design choices at length later on.

The first step of our decoupled synthesis addresses the safety components of the individual objectives.
Given a graph $G$ and an objective $\varphi$, we can use standard procedure to decompose $\spec$ into safety and liveness components, writing $\spec = \specsafe \cap \speclive$, where $\specsafe$ requires some unsafe vertices to be avoided at all time, and $\speclive$ requires some good event to eventually occur.
We show that $\specsafe$ cannot, in general, be guaranteed under stochastic scheduling without additional communication. To this end, we propose a protocol in which agents communicate their sets of safe actions at each step, and the scheduled agent is required to select an action from the intersection of these sets. Under mild assumptions, we prove that this protocol suffices to ensure satisfaction of each safety objective at all times.
Once safety is enforced, we remove all unsafe edges from the graph and exclusively focus on the liveness objectives in the resulting subgraph. For clarity of exposition, we assume this graph is strongly connected, although this assumption is not essential and will be relaxed later.

We first study the existence of %a {\em universal composition framework}; namely, 
a stochastic scheduler that composes \emph{any} collection of policies in a way that the composition almost-surely satisfies all liveness  objectives. 
Somewhat surprisingly, we show that, already for B\"uchi objectives in strongly-connected graphs, such ``universal'' schedulers does not exist. We sketch the intuition in the following example, and details can be found in Thm.~\ref{thm:no univ buchi scheduler}.

%and two B\"uchi objectives, such universal frameworks do not exist.The following example sketches the intuition, the formal treatment is in 

\begin{example}
\label{ex:non-existence-universal-sched}
Consider the graph depicted in \ref{subfig:Rightintro}. We describe two policies $\pi_1$ and $\pi_2$ respectively for the B\"uchi objectives $\varphi_1$ and $\varphi_2$ given by target vertices $\{b_1\}$ and $\{b_2\}$. 
Consider an exponentially increasing sequence, e.g., $n_0 = 1$ and $n_j = 2 \cdot n_{j-1}$, for $j \geq 1$.
Define $\pi_1$ as follows, and $\pi_2$ is dual. When not at $a_1$, proceed towards $b_2$, and when at $a_1$ at time $t$, if $t = n_j$, for $j \geq 0$, proceed to $b_1$, otherwise stay at $a_1$. That is, $\pi_1$ spends exponentially increasing time at $a_1$; visits to $b_1$ become exponentially rare. Clearly, if $\pi_i$ operates alone, it satisfies $\varphi_i$, for $i \in \set{1,2}$. 
Consider a scheduler that at each point in time, chooses which policy acts uniformly at random.
Intuitively, the probability that a random walk visits $a_i$, for $i \in \set{1,2}$, precisely at a time $n_j$, for some $j \geq 1$, is exponentially decreasing, which implies that the probability of visiting $b_i$ infinitely often tends to a constant, thus $\varphi_1 \wedge \varphi_2$ is {\em not} satisfied almost surely. In Thm.~\ref{thm:no univ buchi scheduler}, we generalize the construction to any stochastic scheduler. \qed
\end{example}

%The following example shows, however, that a deterministic universal scheduler does not exist. Somewhat surprisingly, we show that even a stochastic scheduler does not guarantee almost sure satisfaction for B\"uchi objectives. 

%\smallskip
%\noindent\textit{Coordinating policies via convention.}
Example~\ref{ex:non-existence-universal-sched} shows that if the agents are allowed to choose any unrestricted policies, there exist situations when the composition will violate the individual objectives.
We submit that if all the agents are made to choose policies according to some suitable \emph{convention}, the composition will be guaranteed to satisfy all objectives, almost surely.
For example, for B\"uchi objectives, we show that the convention of choosing only finite-memory policies guarantee the fulfillment of all individual objectives.
We formalize conventions as mappings from classes of objectives (B\"uchi, co-B\"uchi, etc.) to restrictions on policies (finite-memory, etc.).
Importantly, the restrictions need to be finalized \emph{before} seeing the specific graph and objectives at hand, setting our framework apart from the joint design of coordinated policies.
This way, the modularity requirement is achieved: agents can independently design their policies based on the respective conventions, and addition, modification, or removal of policies do not affect the rest of the policies.

Conventions can be intuitively explained using a simple metaphor, whose inspiration comes from the work of Lewis~\cite{lewis2008convention}.
Suppose two friends want to meet, but they are unable to talk to each other and arrange the location.
Convention dictates that each of them should go to the place where they most often meet, and if both parties follow the same convention, they will successfully meet.
In essence, convention is a way of inducing coordination when direct communication is not possible.
Coming back to our framework, conventions help agents achieve their objectives without directly synchronizing their policies with each other.

% the agents' choices of carefully crafted infinite-memory policies lead to a composition that violates the B\"uchi objectives. 
%In general, we will see that, depending on the class of objectives, the policies will need some form of coordination to guarantee that the composition achieves the desired outcome for all.
%However, from the modularity standpoint, we bar the agents from coordinating the design of their policies.
%
%
%This motivates restricting the agents' choice of policies, which we study in general as a {\em convention} for choosing policies. A convention formalizes offline coordination among the agents. We seek sufficient conventions that are fixed {\em before} the graph and the objectives are known, such that when all agents follow the convention, the composition is guaranteed to almost-surely satisfy all objectives. 
%For B\"uchi objectives, we show that a very weak convention suffices: if all agents choose finite-memory B\"uchi policies (e.g., an agent applies reinforcement learning to construct their policy), then all objectives are fulfilled almost surely when composing with any stochastic scheduler. 
%Note that conventions go hand in hand with the features we list above; namely, each agent can still choose (construct) a policy independently, preferably with no knowledge of the number of agents in the system nor their objectives, thus both changes to an objective and addition or removal of objectives affect only the relevant policies. 

% -------- NEW ----------
\begin{center}
{\em Our goal is to establish minimally restrictive conventions\\ for different classes of $\omega$-regular objectives.}
\end{center}
% ---------------------

%\paragraph{Conventions for co-B\"uchi objectives.}
While the convention for B\"uchi objectives was merely to use finite-memory policies, the same does not suffice anymore for co-B\"uchi objectives, as we illustrate using the following example.

\begin{example}
\label{ex:coBuchi}
Consider the graph that is depicted in~\ref{subfig:Leftintro} and two co-B\"uchi objectives: $\eventually\always\, \ell$, requiring eventually staying in $\ell$, and $\eventually\always\, (\ell \vee r)$, requiring eventually staying in either $\ell$ or $r$. Observe that for $i \in \set{1,2}$, following the finite-memory policy $\pol_i$ (depicted also in \ref{subfig:Leftintro} in dashed red and dotted blue) results in a path that satisfies the respective objective, but a random composition of the two visits both $t_1$ and $t_2$ infinitely often, almost surely, and thus violates both objectives. \qed
\end{example}

The challenge behind co-B\"uchi convention is that in order to satisfy a conjunction of co-B\"uchi objectives, all policies---that are oblivious to each other's target vertices---must eventually converge into a common cycle that is good for everyone. 
%The challenge in co-B\"uchi objectives stems from policies being chosen without knowing the other objectives. As a result, the policies need to figure out at runtime which edges are safe for all without direct communication among the policies. 
We describe a convention for co-B\"uchi objectives in strongly-connected graphs. Given a co-B\"uchi objective $\varphi_i$, Agent~$i$ chooses a policy $\pol_i$ that operates as follows. (1)~Randomly choose a lasso-shaped path $\rho^i = \rho_i \cdot C_i^\omega$, where $C_i$ is a cycle that satisfies $\varphi_i$. This is Agent~$i$'s guess of the common path everyone is trying to follow. (2)~If scheduled at vertex $v$: choose the successor $v'$ of $v$ on $\rho^i$. (3)~If not scheduled: observe the next vertex $u$, and if $u \neq v'$, a conflict is observed: return to Step~(1) and sample a new path.
\begin{example}
\label{ex:coBuchi}
Consider again the objectives over the graph depicted in \ref{subfig:Leftintro} $\varphi_1=\eventually\always\, \ell$ and $\varphi_2=\eventually\always\, (\ell \vee r)$, the initial position is $t_1$, and let the choice of lasso paths be $\rho^1 = t_1 \ell^\omega$ and $\rho^2 = t_1 t_2 r^\omega$. Observe that $\rho^i$ satisfies $\varphi_i$ for each $i$. Suppose that  $\pol_1$ is scheduled and proceeds to $\ell$. Thus, $\pol_2$ observes a conflict (it expects the next vertex to be $t_2$), and randomly chooses an alternative lasso path that satisfies $\varphi_2$. 
Correctness follows from a measure-theoretic argument: If there exists a common path that satisfies all co-B\"uchi objectives, then there also exists a lasso-shaped path that does the same. 
Random scheduling implies that eventually and almost surely the policies will stabilize in a lasso-shaped path where no policies will change their guesses anymore. 
\qed
\end{example}

%Finally, observe that runtime communication is minimized: the only information a policy receives is whether it is scheduled or not, and the next vertex on the path. Advantages include seamless changes, addition, or removal of objectives, as well as applicability for devices with limited computational capabilities. 

We develop sufficient conventions for major classes of $\omega$-regular objectives, including safety, B\"uchi, co-B\"uchi, and parity, and consider path planning problems on both strongly connected and general graphs. 
As the objectives become more complex, the conventions become more demanding, as expected.
Although our conventions are simple to describe, their correctness proofs involve intricate measure-theoretic arguments.
With this, we create the foundation of a first-of-its-kind modular design framework for $\omega$-regular decision-making.

% ------ Removed based on Tom's comment ----------
%\paragraph{A convention for parity objectives.} 
%We generalise the ideas of the co-B\"uchi convention. Note that satisfying a conjunction of parity objectives is no longer possible with a lasso path. Each policy now randomly chooses (guesses) an {\em internal} memoryless policy that all other policies follow. Once a policy is scheduled, it follows the policy it guessed for itself, and if not scheduled, it observes the move of the scheduled policy to see whether it matches the guess for that policy. If the guess is incorrect, new internal policies are chosen. This convention requires policies with increased scheduling knowledge: each policy knows which policy was scheduled at each point in time. 
% -----------------------------------------------

\paragraph{Design choices.}
We discuss the design choices of our framework. 
{\em (i)~Independent policy choice.} 
In all our conventions, an agent chooses a policy without knowing the objectives of the other agents. This enables robustness: when objective $\varphi_j$ changes, Agent~$i$ does not need to change their policy. 
{\em (ii)~Minimal runtime coordination.} 
For B\"uchi objectives, conventions choose traditional policies (later called {\em path-aware policies}) that use no runtime coordination (see $\pi_1$ and $\pi_2$ in Example~\ref{ex:deterministic versus stochastic scheduler}).
%that make choices only on the history of vertices traversed). 
This means that agents can apply off-the-shelf tools to design policies, and different agents can even apply different techniques, e.g., reinforcement learning or formal techniques. For co-B\"uchi objectives, policies know which time steps they are scheduled in.  Importantly, such policies are {\em not} aware of the number of agents, meaning that additional objectives can be added (offline or at runtime) without changing the policy. 
{\em (iii)~Graph-independent schedulers.} 
%Our schedulers are only aware of the number of agents, which is arguably a minimal requirement of schedulers. 
To motivate this, we return to Example~\ref{ex:deterministic versus stochastic scheduler}. Intuitively, composition using the deterministic scheduler $\sched_{\mathsf{d}}$ failed because none of the agents were getting enough time to reach their target. 
A na\"ive fix simply schedules each policy \emph{three} times in a row. While this does solve the problem at hand, it is an ad hoc solution that %relies on knowledge of the graph and the objectives and 
will break if either graph or objectives change. Alternatively, one could consider a \emph{nondeterministic} scheduler and construct policies that fulfill their objectives no matter how the nondeterminism is resolved. 
Unfortunately, this is too restrictive from the policy design standpoint: in Example~\ref{ex:deterministic versus stochastic scheduler}, given the policies $\pi_1$ and $\pi_2$, a nondeterministic scheduler can act like an adversary and find an interleaving of the policies (like the one of $\sched_{\mathsf{d}}$) such that the path keeps cycling between $t_1$ and $t_2$ and never visits the targets.
Defeating such an adversarial scheduler is possible, but would require the policies to take joint actions (like $\pi_1$ and $\pi_2$ making a clockwise traversal of the graph in Example~\ref{ex:deterministic versus stochastic scheduler}), which would preclude the key aspect of modular design.
{\em (iv)~Simple schedulers}. We see it as a feature of our framework that our schedulers are simple and can thus run on devices with limited processing capabilities like robots or drones. 

\subsection*{Related Work}
Computing policies for multi-objective decision-making problems is a well-studied problem, both for probabilistic~\cite{basset2015strategy,basset2018compositional,chatterjee2018combinations} and non-probabilistic~\cite{roijers2013survey} systems.
Almost all works produce a single, monolithic policy that fulfills all objectives. 

Decoupling has been considered in other fields with the same motivation as us; a framework that is modular, robust, and separates design concerns. For example, decoupling exploration from exploitation has been studied in RL~\cite{SCHA22}, {\em behavioral programming} is a decoupled approach to programming~\cite{HMW12}, decoupling has been applied in control~\cite{GB05}, and game theory~\cite{HM03}. 

We compare with two approaches for decoupled design that are closest to our work. 
The first kind ``over-approximates'' the local policies using the so-called \emph{strategy templates}~\cite{anand2023synthesizing}, which abstractly represent sets of policies that would fulfil the objectives.\footnote{Strategy templates are originally developed on turn-based game graphs, which are generalizations of the plain graphs that we use in this work.}
When these local strategy templates are composed, we obtain a strategy template for the conjunction of all objectives, from which a concrete policy for the overall problem can be extracted.
Unfortunately, the composition of strategy templates is computationally intensive, which creates substantial computational overhead and can impact the runtime performance of the composed policy.

The second kind of modular design uses auctions as a scheduler, and is called {\em auction-based scheduling} (ABS)~\cite{avni2024auction}: 
Each policy starts with a budget, and at each step, an auction is held, where the policies use their available budgets to bid for the privilege of selecting the next action.
{\em Bidding policies} for ABS are found by independently solving a two-player zero-sum {\em bidding game}~\cite{lazarus1999combinatorial,AHC19}. 
We point to two key advantages of our approach over ABS. First, our approach allows an arbitrary number of objectives whereas ABS for multiple objectives is currently not possible due to the dependency on multi-player bidding games, which, to best of our knowledge, have not yet been studied. 
Second, ABS is a sound but incomplete method: it is possible that there is a path that satisfies both objectives, but there is no decoupled solution for ABS. In fact, most instances with a pair of co-B\"uchi objectives cannot be decoupled using ABS. On the other hand, our framework is guaranteed to successfully decouple any collection of parity objectives on strongly-connected graphs, assuming, of course, that they have an intersection.

In several graph games, ``mixed'' or \emph{randomised policies} are used either out of necessity for winning against the opponent~\cite{de2000concurrent} or to reduce the memory requirements of the policies~\cite{chatterjee2004trading}.
%A randomised policy is one that selects probability distributions over actions at each step, instead of a deterministic action.
Our randomised scheduling of local policies can be viewed as randomised policies, where the randomisation over actions is facilitated by the scheduler.
However, there is an important difference: 
usually, randomised policies are designed centrally for the entire objective, whereas our goal is to design isolated local policies independently from the scheduler.
%This creates new technical opportunities and challenges, and many results on the capabilities and limitations of randomised policies do not extend to our setting.

The classical distributed view of reactive synthesis~\cite{pneuli1990distributed,kupermann2001synthesizing,finkbeiner2005uniform} uses a fundamentally different model. % of the problem.
It is assumed that the actions chosen by each local policy affect a separate set of variables, so that all local policies can be deployed in parallel.
In other words, classical distribution synthesis does not require scheduling of policies, and asks how to design local policies whose parallel composition fulfills all objectives.

%\KM{Putting the examples here at the moment...}
%
%\begin{example}[Universal scheduler for B\"uchi objectives with dense, bounded hitting time policies.]
%	Consider the game in Fig.~\ref{fig:screenshotTempBuchi}, where the policy indices are $\set{1,2}$.
%	Suppose the scheduler $\sched$ is the uniform scheduler as described in Cor.~\ref{cor:Buchi+dense+bounded hitting+uniform scheduler}.
%	We intuitively explain why this uniform scheduler is also universal.
%	Consider a pair of finite memory policies for the two B\"uchi objectives, where policy $\pol_1$
%	
%\end{example}
\section{Preliminaries}
%!TEX root=main.tex

We use $\mathbb{N}$ to denote the set of natural numbers $\{0,1,2,\dots\}$. For two natural numbers $i,j$ with $i<j$, we use $[i; j]$ to denote the set $\{i,i+1,\dots,j\}$ consisting of natural numbers that are at least $i$ and at most $j$. 
Sometimes, for a natural number $i$, we use $[i]$ to denote the set $[1; i]$. 

\medskip
\noindent\textit{Graphs, and paths. }  A graph is an ordered triple $\G = (\V, \E, \vinit)$, where $\V$ is a finite set of vertices, $\E \subseteq \V \times \V$ is a set of directed edges, and $\vinit \in \V$ is the initial vertex. 
A path in $\G$ is a (finite or infinite) sequence of vertices $v_0 v_1 \ldots$ such that $v_0 = \vinit$ and $(v_i, v_{i+1}) \in \E$ for every valid index $i$. We write $\pathsfin(\G)$ and $\pathsinf(\G)$ for the sets of finite and infinite paths of $\G$, respectively. 
By convention, $\pathsfin(\G)$ includes the path of length $0$.

\medskip
\noindent\textit{Objectives.} 
An \emph{objective} $\spec$ in $\G$ is a set of infinite paths in $\G$. 
We represent objectives defined using predicates over the vertices of the given graph. We consider the following \emph{families} of objectives. 
\begin{description}
\item[Reachability.] is specified by a set of \emph{target} vertices $T \subseteq V$, and is the set of every path that eventually visits at least one vertex in $T$. 
\item[Safety.] is specified by a set of \emph{safe} vertices $S \subseteq V$, and is the set of every path always remaining within the vertices in $S$. 
\item[Parity.] is specified by a \emph{colouring} function $\kappa: V\mapsto \mathbb{N}$, which associates each vertex a natural number called its \emph{colour}. An infinite path is in the parity objective iff the largest colour that appears infinitely often is even.
\item[B\"uchi.] is a special case of parity in which vertices are only coloured by $1$ and $2$, where the latter are called \emph{B\"uchi} vertices. A path is in the B\"uchi objective iff B\"uchi vertices are visited infinitely often. 
\item[Co-B\"uchi.] is a special case of parity in which vertices are only coloured by $0$ and $1$, where the latter are called \emph{co-B\"uchi} vertices or ``bad'' vertices. A path is in the co-B\"uchi objective iff it visits co-B\"uchi vertices only finitely often. 
\end{description}

We will divide objectives into two complementary families, namely safety and liveness, defined below.
While safety can be defined using safe vertices as above, the following definition uses an alternate, semantic variant.

\begin{definition}[Safety versus liveness objectives]
	Let $\G$ be a given graph.
	A given objective $\spec$ in $\G$ is called \emph{liveness} if every finite path in $\G$ can be extended into an infinite path in $\spec$.
	Dually, a given objective $\spec$ in $\G$ is called \emph{safety} if every infinite path $v_0v_1\ldots\notin\spec$ has a finite prefix $\rho = v_0\ldots v_i$ such that for every $j\geq i$, all infinite extensions of the path $v_0\ldots v_j$ is not in $\spec$.
\end{definition}

\medskip
\noindent\textit{Almost-sure satisfaction.}
%Guy: I think these are standard definitions from probability. It would be good to keep in a paper like this. 
%\KM{Crazy thought: do we need ``regular'' policies at all? Why not totally remove the following paragraph?}
Let $\mathcal{D}(\V)$ denote the set of all probability distributions over $\V$, and given $d\in \mathcal{D}(\V)$, we will write $\supp(d)$ to denote the support of $d$. 
Consider a function $f:\pathsfin(\G)\to \mathcal{D}(\V)$ that, given a finite path $\rho$ ending at a vertex $v$, assigns a probability distribution whose support is contained in the set of successors of $v$. 
We extend $f$ to infinite paths as follows. 
For a finite path $\rho = v_0\ldots v_k$, the cylinder set spanned by $\rho$, denoted $\Cyl(\rho)$, is the set of all infinite paths whose prefix is $\rho$. 
Define the probability distribution over cylinder sets based on $f$ as $\Pr^{f}(\Cyl(\rho)) = \prod_{i=0}^{k-1} f(v_0\ldots v_i)(v_{i+1})$.
The function $\Pr^{f}$ is a pre-measure over the set $\pathsinf(\G)$, which extends to a \emph{unique} probability measure over $\pathsinf(\G)$ by applying the Carathéodory’s extension theorem~\cite{Bil95}.
For simplicity, we use the notation $\Pr^{f}$ to also represent the probability measure.
We say that $f$ {\em almost surely} satisfies an objective $\varphi$ if $\Pr^f(\spec) = 1$.

\section{The Decoupled Planning Framework}\label{sec:framework}
We develop a decoupled framework for generating a path that satisfies a collection $\varphi_1,\ldots, \varphi_N$ of objectives. 
The framework has two main ingredients, $N$ individual {\em agents} and {\em a scheduler}.
Each Agent~$i$, for $i \in [1;N]$, is solely interested in fulfilling the objective $\varphi_i$.
To this end, the Agent~$i$ chooses a policy $\pol_i$, and the policies of all agents are composed at runtime by the scheduler.

%of stam

This section is devoted to formalizing this framework. We start with schedulers. 

\begin{definition}[Schedulers: general, fair, deterministic]\label{def:schedulers}
A \emph{scheduler} over $[N]$ is a function $\sched:[N]^*\to\Distribution([N])$.
    The scheduler $\sched$ is called \emph{fair} if there exists $\epsilon > 0$ such that for every $u\in [N]^*$ and for every $n\in [N]$, $\Pr(n = \sched(u)) > \epsilon $.
    A scheduler is called {\em deterministic} if it always chooses Dirac distributions.
\end{definition}

We reiterate the advantages of stochastic schedulers mentioned in the introduction. First, they are completely agnostic of the underlying graph and the objectives, while not obstructing policies to fulfill their objectives.
Moreover, stochastic schedulers are lightweight and easy to implement.
In fact such schedulers are often used in concurrent programs, where multiple programs (analogous to our policies) that are competing for a resource (being scheduled, analogous to our policies being able to execute their actions) are allocated the resource randomly.
Henceforth, schedulers will by default be stochastic.

%Guy: It's one line and too close to the deadline. I'm keeping it here. 
%\KM{Can't we put the deterministic schedulers where it is used? Just to not have definition overload. Plus it's a bit odd if you look the definitions of fair and deterministic together: no deterministic scheduler is fair according to these definitions.}

\noindent\textbf{Semantics of schedulers.}
Every sequence from $[N]^*\cup [N]^\omega$ is called a \emph{schedule} which is either finite or infinite.
Every scheduler $\sched$ induces a probability measure over the infinite schedules as follows.
Given a finite sequence $\theta = p_0\ldots p_k\in [N]^*$, we define $\Pr^\sched(\Cyl(\theta)) = \prod_{i=0}^{k-1} \sched(p_0\ldots p_{k-1})$, which is a pre-measure that extends to a unique measure---also denoted as $\Pr^\sched$---over the set  $[N]^\omega$ by applying Carathéodory’s extension theorem~\cite{Bil95}.\qed

\smallskip
We formalize the second ingredient of our framework, the policies that the agents choose, where %of the $N$ agents, given a fixed scheduler that is unaware of the graph and the objectives.
we define policy types with increasing knowledge of scheduling choices. 
We introduce some notation first. 
Given two alphabets $A$ and $B$, and a pair of finite words of equal length $w_A = a_0\ldots a_k\in A^*$ and $w_B = b_0\ldots b_k\in B^*$, we define their element-wise cross-product as: $w_A\otimes w_B \coloneqq (a_0,b_0)\ldots (a_k,b_k)\in (A\times B)^*$.
Furthermore, we lift this cross-product to sets of sequences: given $W_A\subseteq A^*$ and $W_B\subseteq B^*$, define $W_A\otimes W_B\coloneqq$ $\{w_A\otimes w_B\mid w_A\in$ $W_A, w_B\in W_B\}$.

\begin{definition}[Policies augmented with scheduling information]
Let $\sched$ be a scheduler over $[N]$. %We use $\id\in [N]$ to refer to the agent that is scheduled. 
	\begin{description}
\item[Path-aware.] policies coincide with the traditional definition of policies and have no scheduling knowledge, they are of the form $\pol \colon \pathsfin(\G) \to \mathcal{D}(\V)$. We denote the set of path-aware policies by $\PolP$.

\item[Scheduled-aware.] 
policies keep track of which past time points they have been scheduled, and all choices made by it in the past.
A scheduled-aware policy is a \emph{partial} function of the form $\pol\colon\pathsfin(\G)\otimes \set{\top,\bot}^*\otimes \V^* \to \mathcal{D}(\V)$, mapping the path $\rho=v_0\ldots v_k\in \pathsfin(\G)$, the history  $\theta=q_0\ldots q_{k-1} \in \set{\top,\bot}^*$ of time points in which it was scheduled, where ``$\top$'' means ``scheduled'' and ``$\bot$'' means not scheduled, and the history of choices $\gamma=w_1\ldots w_{k}\in \V^*$ made by it, to the distribution $\pol(\rho,\theta,\gamma)$.
		The distribution $\pol(\rho,\theta,\gamma)$ is defined iff $\rho$, $\theta$, and $\gamma$ are consistent, i.e., for every $i\in [k]$ with $q_{i-1} = \top$, $v_{i}=w_i$.
%		\SS{Shouldn't it be "if \(q_{i-1} = \top\) then \(v_i = w_{i-1}\)? Moreover, the length of \(\rho, \theta\), and \(\gamma\) should be same. }
		The set of all full scheduled-aware policies will be denoted as $\PolS{}$.

\item[Full history-aware.] policies enrich scheduled-aware policies by additionally including the indices of scheduled agents during times it was unscheduled (instead of just keeping track of ``$\bot$'' as scheduled-aware policies do).
Without loss of generality, assume that the indices of the other policies are $[2;N]$.
A full history-aware policy is a \emph{partial} functions of the form $\pol\colon\pathsfin(\G)\otimes \set{\top,2,\ldots,N}^*\otimes \V^* \to \mathcal{D}(\V)$, mapping the path $\rho\in \pathsfin(\G)$, the history of scheduling decisions $\theta\in \set{\top,2,\ldots,N}^*$, and the history of choices $\gamma\in \V^*$ made by $\pol$, to the distribution $\pol(\rho,\theta,\gamma)$.
The distribution $\pol(\rho,\theta,\gamma)$ is defined under the same conditions as for scheduled-aware policies.
		The set of all full history-aware policies with $N$ agents will be denoted as $\PolFH{N}$.
	\end{description}
For each policy, regardless of its type, for every given path $\rho v\in\pathsfin(\G)$, the support of the output distribution over vertices must be a subset of $\supp(v)$.
\end{definition}

%\KM{Can we give examples of these three types of policies based on the Figures in the intro?} Guy: Maybe for the next submission. 

\begin{remark}[Why keep track of past choices?]
One may na\"ively think that it is redundant that policies in $\PolS{}$ and $\PolFH{N}$ keep track of their own past choices, alongside the whole path history, because the old histories could be used to simulate the old choices on demand.
This will not work since the policies are stochastic, and executing them on the same history will generate different outcomes on different instances with positive probability.
\end{remark}

\begin{remark}[Comparing decoupling strength]
Decoupling is stronger when using a family of policies with weaker scheduling information. 
Since path-aware policies coincide with the traditional definition of policies, an agent can use any off-the-shelf tool to construct a policy for their objective. This is no longer the case with scheduled-aware and full history-aware policies. Yet, an appealing feature of scheduled-aware policies is their independence of how many policies are interacting. 
This enables, for example, ``patching'' a deployed system by removing or adding an objective at runtime. Indeed, even when $N$ changes, no update is needed to an individual policy. This is no longer the case with full history-aware policies. We point out that decoupling with full history-aware policies still provides modularity: as long as the number of objectives is fixed, changes to objectives at runtime is possible by changing only the relevant policies. 
\end{remark}

We define the \emph{composition} of a scheduler and a collection of full history-aware policies. This generates a probability distribution over the infinite paths generated by a random interleaving of the individual local policies.
\begin{definition}[Composition]\label{def:composition}
	Suppose we are given a graph $\G$, a scheduler $\sched$ for a set of $N$ agents, and a set of policies $\pol_1,\ldots,\pol_N$ where for every $i\in [N]$, $\pol_i \in \PolFH{N}\cup \PolS{N}\cup \PolP$.
	The composition is the tuple $\implementation = (\G,\pol_i,\dots,\pol_N,\sigma)$.
\end{definition}
Each composition defines a probability distribution over the paths generated, which we computed step-by-step below. Such a distribution depends on the type of scheduler, hence we define it carefully below.

\smallskip
\noindent\textbf{Semantics of composition.}
At each time point $i$, when the current path is $v_0\ldots v_i \in \pathsfin(\G)$, each Agent $j$ randomly selects the next vertex $w_{i+1}^j$ using its policy, giving us a tuple $(w_{i+1}^1,\ldots,w_{i+1}^N)$ of selections.
Out of this tuple, the choice of the scheduled agent $p_i$ is used to extend the current path to $v_0\ldots v_{i+1}$, where $v_{i+1}=w_{i+1}^{p_i} $, and the process continues.
For a given fixed schedule $p_0\ldots p_{k-1}\in [N]^*$, we define the probability of the sequence of selections being $(w_1^1,\ldots,w_1^N) \ldots (w_k^1,\ldots,w_k^N)$ given as:
\begin{align}\label{eq:implementation:conditional probability}
	\Pr^{\implementation}\left((w_1^1,\ldots,w_1^N) \ldots (w_k^1,\ldots,w_k^N) \mid (p_0,\ldots p_{k-1})\right) \coloneqq \prod_{j=1}^N\prod_{i=0}^{k-1} D(w_{i+1}^j),
\end{align}
where $D(w_{i+1}^j)$ represents the probability that at time $i+1$, the policy $\pol_j$ selected the vertex $w_{i+1}^j$, whose value is defined based on the policy type as below.
In the above, the path generated by the schedule $p_0,\ldots p_{k-1}$ from the sequence $(w_1^1,\ldots,w_1^N) \ldots (w_k^1,\ldots,w_k^N)$ is $v_0v_1\ldots v_k$ where $v_0=\vinit$ and for each $i\in [k]$, $v_i = w_i^{p_{i-1}}$.
The set of all policies are $\set{\pol_j}_{j\in [N]}$.
Then,
\begin{multline*}
	D(w_{i+1}^j) = \\
		\begin{cases}
			\pol_{j}\left(v_0,\ldots v_i\right)	&	\pol_{j}\in \PolP,\\
			\pol_{j}\left((v_0\ldots v_{i}),(q_0\ldots q_{i-1}),(w_1^j\ldots w_{i}^j)\right),\, \forall t\;.\;q_t=\top\Leftrightarrow p_t=j	&	\pol_{j}\in \PolS{N}, \\
			\pol_{j}\left((v_0\ldots v_{i}),(p_0\ldots p_{i-1}),(w_1^j\ldots w_{i}^j)\right)	&	\pol_{j}\in \PolFH{N}.
		\end{cases}
\end{multline*}
%------- deterministic policies -------
%\KM{Move the following paragraph to somewhere else, too much is happening already...}
%If the policies are all deterministic, then we obtain a unique path $\rho_\theta = v_0\ldots v_k$ given the schedule $\theta = p_0\ldots p_{k-1}$, where for each $i\in [0;k-1]$, $v_{i+1} = \pol_i(v_0\ldots v_i)$.
%As expected, we write $\Pr^{\implementation}(\rho \mid \theta) = 1$ if $\rho=\rho_\theta$, and $\Pr^{\implementation}(\rho \mid \theta) = 0$ otherwise.
% -----------------------------------

Suppose we are given a finite path $\rho=v_0\ldots v_k\in \pathsfin(\G)$ and a schedule $\theta = p_0\ldots p_{k-1}\in [N]^*$.
We introduce the set $\Sigma_{v_0,\rho,\theta} \subseteq \V^{N\times k}$ representing the set of all sequences of vertices generated by all agents such that the schedule $\theta$ would produce the path $\rho$ starting at the initial vertex $v_0$.
Formally, $\Sigma_{v_0,\rho,\theta} \coloneqq \set{(v_1^1,\ldots,v_1^N) \ldots (v_k^1,\ldots,v_k^N) \mid \forall i\in [k]\;.\;v_{i}=v_i^{p_{i-1}}}$.
Now we combine $\Pr^{\implementation}(\rho \mid \cdot)$ with the probability measure $\Pr^\sched$ over sequences of policy indices, and write $\Pr^{\implementation}(\Cyl(\rho))$ as
\begin{multline}\label{eq:implementation:probability over cylinder sets}
 \sum_{\substack{\theta\in [N]^{k-1} \\ {(v_1^1,\ldots,v_1^N) \ldots (v_k^1,\ldots,v_k^N)\in \Sigma_{v_0,\rho,\theta}}}}
 \Pr^{\implementation}((v_1^1,\ldots,v_1^N) \ldots (v_k^1,\ldots,v_k^N) \mid \theta)\cdot\Pr^{\sched}(\theta),
\end{multline}
 which is a pre-measure and leads to a unique measure---also written $\Pr^{\implementation}$---over the set of all infinite paths by applying the Carathéodory’s extension theorem~\cite{Bil95}.
We will write $\implementation \modelsAS \spec$ iff $\Pr^{\implementation}(\spec)=1$.\qed

Following is the decoupling problem that we undertake.

\begin{mybox}[width=\textwidth,colback={blue!10!white},title={Problem statement (informal): decoupled planning}][
	Fix a fair, stochastic scheduler $\sched$.
	Suppose $\G$ is a given graph and $\spec_1, \ldots ,\spec_N$ are the objectives.
		Independently design the policies $\pol_1,\ldots,\pol_N$ in a way that $\implementation (\G,\pol_1,\ldots,\pol_N,\sched)\modelsAS \spec_1\land \ldots \land\spec_N$.
\end{mybox}

Towards this goal, we start by stating why a fair \emph{deterministic} scheduler would not suffice in achieving the desired decoupling, which formalizes the informal description in Example~\ref{ex:deterministic versus stochastic scheduler}.

\begin{restatable}[No deterministic scheduler can be universal]{theorem}{NoDetSched}
\label{prop:NoDetScheduler}
    For every deterministic scheduler $\sched$, there exist a graph \(\G\) and reachability
    objectives $\spec_1$ and $\spec_2$ with $(\spec_1\cap\spec_2) \neq \emptyset$ and deterministic policies  $\pol^\sched_1$ and $\pol^\sched_2$ that respectively satisfy $\spec_1$ and $\spec_2$, but 
    %the unique path generated by    $(\G,\pol_1^\sched,\pol_2^\sched,\sched)$    does not satisfy $\spec_1\cap\spec_2$. 
     $(\G,\pol_1^\sched,\pol_2^\sched,\sched) \not \models \spec_1 \wedge \spec_2$. 
\end{restatable}

The high level idea (slightly over-simplified) of the proof appears in Example~\ref{ex:deterministic versus stochastic scheduler}, and a rigorous proof can be found Appendix~\ref{app:NoDetScheduler}.
This limitation of (extremely) fair deterministic schedulers inspired us to use stochastic schedulers with the stronger probabilistic fairness~\cite{de1998formal} properties.

\section{Building Support for Safety Components of\\ $\omega$-Regular Objectives}\label{sec:safety}
%Our compositional framework as defined in Section~\ref{sec:safety} unfortunately falls short even if only one of the specifications is safety.
We start by demonstrating the challenge of incorporating safety. 
%through an example, and then propose a solution that achieves soundness with minimal communications and under mild conditions on safety specifications.

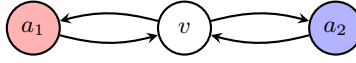
\begin{figure}
	 	\centering
	 	\begin{tikzpicture}[>=stealth,thick,node distance=2cm]
	 		
	 		\tikzstyle{state}=[circle,draw,minimum size=7mm,inner sep=0pt]
	 		\tikzstyle{target1}=[circle,draw,minimum size=7mm,inner sep=0pt,fill=red!30]
	 		\tikzstyle{target2}=[circle,draw,minimum size=7mm,inner sep=0pt,fill=blue!30]
	 		
	 		% Nodes
	 		\node[state] (v) {$v$};
	 		\node[target2, right of=v] (a2) {$a_2$};
	 		\node[target1, left of=v] (a1) {$a_1$};
	 		
	 		%Edges
	 		\draw[->] (v) edge[bend left=15] (a2);
	 		\draw[->] (a2) edge[bend left=15] (v);

	 		\draw[->] (v) edge[bend right=15] (a1);
	 		\draw[->] (a1) edge[bend right=15] (v);
	 		
	 		%\draw[->] (v) edge[loop below] ();

	 	\end{tikzpicture}
 			\caption{A graph with a safety objective $G (\neg a_1)$ and a B\"uchi objective $GF a_2$.}\label{fig: determinB}
	 \end{figure}
	 
\begin{example}
Consider the graph depicted in Figure~\ref{fig: determinB}, and two objectives, the safety objective $\varphi_1$ requires ``avoid the Red vertex $a_1$'' and a B\"uchi objective $\varphi_2$ requires ``visit the Blue vertex $a_2$ infinitely often''. Suppose that Agent~$2$ chooses the policy $\pi_2$ that alternates between selecting $a_1$ and $a_2$ at $v$. Note that if $\pi_2$ is always scheduled, then $\varphi_2$ is satisfied. Observe that no matter which policy Agent~$1$ chooses, any stochastic scheduler that schedules $\pi_2$ at $v$ with probability at least $\epsilon>0$ leads to a random walk that visits $a_1$ almost surely and thus violates $\varphi_1$. 
Finally, note that the path $(v, a_1)^\omega$ satisfies $\varphi_1 \wedge \varphi_2$. \qed
\end{example}

As demonstrated in Examples~\ref{ex:deterministic versus stochastic scheduler} and~\ref{ex:coBuchi}, liveness objectives allow for flexible decoupling. Safety, however, is rigid. 
%The fact that our scheduling-based implementations cannot support safety is not surprising. This is because, under random scheduling, the opponents can force visit to every reachable vertex---including the unsafe ones---in the graph with nonzero probability, thereby not fulfilling safety almost surely.
Indeed, it is well known that almost-sure and sure satisfaction coincide for safety~\cite{de2000concurrent}. That is, in order to satisfy safety, a composition needs to satisfy it deterministically on \emph{every} path. %---which is not possible with the scheduler defined in Definition~\ref{def:schedulers}.
This requires communication, which we establish by augmenting the policies with a safety-preserving component defined as follows. 

\smallskip
\noindent\textit{Maximally permissive policies.}
A \emph{maximally permissive policy} for a safety specification $\varphi$ is $\chi: V \rightarrow 2^V$. Intuitively, for a vertex $v$, the set $\chi(v)$ represents safe continuations from $v$. More formally, a path $\rho = v_0, v_1,\ldots$ such that $v_{i+1} \in \chi(v_i)$ satisfies $\varphi$. Maximality of $\chi$ means that for every $v$ and every $u \notin \chi(v)$, there is a continuation from $u$ that violates $\varphi$. 
In order to construct $\chi$, we compute the set of vertices from which $\spec_i$ can be fulfilled, which we  call the \emph{winning region} and denote it by $W$. Then, for every $v \in W$, we define $\chi(v) = \set{u: E(v,u) \wedge u \in W}$. 
Observe that $\chi$ is maximal; proceeding to $u \notin \chi(v)$ violates $\varphi$.

\smallskip
\noindent\textit{Shielded composition.}
We augment the composition defined in Def.~\ref{def:composition} with maximally-permissive policies. Consider policies $\pi_1,\ldots, \pi_N$ and maximally-permissive policies $\chi_1,\ldots, \chi_N$. 
Whenever Agent~$i\in [N]$ is scheduled at vertex $v$, every other Agent~$j\neq i$ sends $\chi_j(v)$ to Agent~$i$. 
Agent~$i$ selects a vertex $\pol_i(v) \in \bigcap_{1 \leq j \leq N} \chi_j(v)$ that is safe for all. 
%locally computes $W\coloneqq \bigcap_{j\neq i} W_j$, and selects an action that is in the intersection of $\pol_i(v)\cap W$.
We will prove that, under mild conditions on the graph and objectives, there is always such a safe vertex to choose. 
%intersection $\pol_i(v)\cap W$ will always be non-empty. This way, if all agents follow this same protocol, every infinite path generated by the composition will fulfill all safety objectives.

\smallskip
\noindent\textit{Decomposition of $\omega$-regular specifications.} 
Consider an $\omega$-regular objective $\varphi_i$ for Agent~$i$. 
It is possible to decompose $\varphi_i$ into a safety and a liveness component  $\specsafe_i$ and $\speclive_i$, such that $\spec_i = \specsafe_i \cap\speclive_i$~\cite{alpern1987recognizing}.
Let $W_i$ be the maximal winning region, $W_i^c$ the complement of $W_i$, and $\reach(W_i^c)$ the objective that contains all paths eventually reaching $W_i^c$.
Define $\speclive_i\coloneqq \spec_i \cup \reach(W_i^c)$.
Observe that $\speclive_i$ is a liveness objective, since no matter what prefix we have seen, we can either stay in $W_i$ and fulfil $\spec_i$, or reach $W_i^c$.
Clearly $\spec_i = \specsafe_i\cap \speclive_i$, and moreover, $\speclive_i$ can be expressed as a parity objective (possibly with additional colours).

We describe how the agents act. Let $\pi_i$ be a policy for Agent~$i$ and let $\chi_i$ be a maximally permissive policy. As before, when Agent~$i$ is scheduled following path $\rho$ that ends at $v$, every other Agent~$j$ sends $\chi_j(v)$ and Agent~$i$ selects $\pi_i(\rho) \in \bigcap_j \chi_j(v)$. 
We call this a \emph{shielded} composition, and denote it by $(\G,\pol_1,\ldots,\pol_n,\sigma)^\sh$. We omit the maximally permissive policies for brevity. 

\smallskip
\noindent\textit{Mutual safety closed.}
Consider an objective $\spec\subseteq V^\omega$, we define $\pref(\spec)\coloneqq \set{u\in V^\star\mid \exists v\in \V^\omega\;.\; uv\in \spec}$ to be the set of all finite prefixes of $\spec$.
We call a pair of objectives $\spec_1$ and $\spec_2$ \emph{mutually safety-closed}, if every path $\path$ in $\G$ fulfilling $\path\in \pref(\spec_1)\cap \pref(\spec_2)$ also fulfils $\path \in \pref(\spec_1\cap \spec_2)$.
By transitivity, we can generalize mutual safety-closure to arbitrary number of objectives.

The mutually safety closed assumption guarantees that every prefix of a path has an extension that satisfies each objective. This ensures that a path that makes a ``wrong step'' can always recover.

\begin{assumption}\label{assuump:safetty closure}
	The objectives $\set{\spec_1,\ldots,\spec_N}$  are mutually safety-closed.
\end{assumption}

	Let $W_i$ be the maximal set of vertices in $\G$ from which the agent $i$ has a policy to fulfil $\spec_i$.
	In other words, $\specsafe_i = \safe(W_i)$.
	Define $W\coloneqq \cap_{i\in [N]} W_i$. 
	By assumption, $W$ is nonempty, since $\vinit\in W_i$ for all $i$.	
	Let $\G[W]$ be the subgraph obtained by removing all the vertices $\V\setminus W$ from $\G$, and by removing all the edges that involve these vertices.
	If Assumption~\ref{assuump:safetty closure} holds, then it can be shown that the subgraph $\G[W]$ is deadlock-free (see Prop.~\ref{Prop:subgraph is dead end free} in Appendix~\ref{appn: safety}), and therefore every path can be extended into an infinite path.

We are now ready to present the main theorem of this section (see its proof in Appendix~\ref{appn: safety}), which suggests the following decoupling framework for $\omega$-regular objectives. The shielded composition that is described above keeps the path in the \emph{strongly-connected} region $G[W]$ that is safe for all agents. This leaves us with liveness, which the policies $\pi_1,\ldots, \pi_N$ take care of. Constructing $\pi_1,\ldots,\pi_N$ is not trivial. In subsequent sections, we develop such constructions 
%only focus on the left side of the implication in Eqn.~\ref{eq:safeety liveness}, 
while assuming that the graph is strongly connected and all objectives are liveness objectives.

\begin{restatable}[Compositional enforcement of the safety components]{theorem}{safetyclosed}\label{thm:safetyclosed}
	Suppose Assumption~\ref{assuump:safetty closure} holds true.
	Then for every fair scheduler $\sched$ and every set of local policies $\pol_1,\ldots,\pol_N$ the following holds:
	\begin{align*}%\label{eq:safeety liveness}
		\left[(\G[W],\pol_1,\dots,\pol_N,\sigma)\modelsAS \bigcap_{i=1}^N \speclive_i\right] \implies
		\left[(\G,\pol_1,\dots,\pol_N,\sigma)^\sh\modelsAS \bigcap_{i=1}^N \spec_i\right].
	\end{align*}
\end{restatable}

%\new{
%In light of Theorem~\ref{thm:safetyclosed}, from this point onward, we only focus on the left side of the implication in Eqn.~\ref{eq:safeety liveness}, and assume that the given graph is strongly connected and all objectives are liveness objectives.}
%\KM{Adjust the definitions of these terms in the preliminaries.}

\section{Decoupling Liveness Objectives via \emph{Conventions}}
%\section{Non-existence of Sound and Universal Protocols}\label{sec:nonexist}
%\section{Non-existence of the Universal Decoupling Framework}\label{sec:nonexist}
Unlike safety objectives, we show that liveness objectives \textit{can} be decomposed without any direct communications among agents. 
However, this does not rule out the need for coordination.
In fact, already for B\"uchi objectives on strongly connected graphs, without coordination the agents may not be able to fulfill their objectives.
This is somewhat surprising, and the intuition was sketched in Example~\ref{ex:non-existence-universal-sched}.
We formally state this result in the following theorem.
 
\begin{restatable}[No stochastic scheduler for uncoordinated B\"uchi policies]{theorem}{NoUnivBuchischeduler}
\label{thm:no univ buchi scheduler}
There is a graph $\G$ such that for any stochastic scheduler $\sched$, there are B\"uchi objectives $\spec_1$ and $\spec_2$ with $(\spec_1\cap\spec_2) \neq \emptyset$ and deterministic policies $\pol^\sched_1$ and $\pol^\sched_2$ that respectively satisfy $\spec_1$ and $\spec_2$, but the composition $(\G,\pol_1^\sched,\pol_2^\sched,\sched) \not \modelsAS \spec_1\cap\spec_2$.     
\end{restatable}

We remark that the proof above also shows that even for a reachability objective, despite using  B\"uchi strategies that ensure that every prefix satisfies the B\"uchi objective, the probability of reaching an accepting state can be bounded by $1/2$ (union bound).
\begin{corollary}[No stochastic scheduler for uncoordinated reachability objectives]
There exists a graph $\G$ such that for every scheduler $\sched$, there exists reachability objectives $\spec_1$ and $\spec_2$ with $(\spec_1\cap\spec_2) \neq \emptyset$ and deterministic policies $\pol^\sched_1$ and $\pol^\sched_2$ that respectively satisfy $\spec_1$ and $\spec_2$, but $(\G,\pol_1^\sched,\pol_2^\sched,\sched) \not \modelsAS \spec_1\cap\spec_2$.     
\end{corollary}

These results motivate us to seek mechanisms that would induce the required coordination among agents, without them requiring to establish direct communication during the design and deployment of their policies.
To this end, we formalize {\em conventions}, which are rules that agents should follow while selecting their policies, and it is guaranteed if every agent follow these conventions, the composition will fulfill all objectives. 
Fig.~\ref{fig:convention} illustrates how conventions are used.

\begin{definition}[Conventions]
Consider a family of objectives $\alpha$, where we focus on $\alpha \in \set{\text{B\"uchi}, \text{co-B\"uchi}, \text{parity}}$. 
A convention $\Conv_\alpha$ for the family $\alpha$ is a mapping that takes a graph $\G$, a scheduler $\sched$, and an objective $\varphi$, and outputs a collection of policies of type $\beta$ that the agent can choose from. We focus on $\beta \in \set{\SAPolFH{N},\SAPolS{N},\SAPolP}$ and design a convention $\Conv_\alpha$ such that for every $\G,\sched$, and $\varphi \in \alpha$, we have $\Conv_\alpha(\G, \sched, \varphi) \subseteq \beta$. 
\begin{itemize}
\item Soundness: A convention $\Conv_\alpha$ is sound for classes of graphs and schedules, if for every graph $\G$ and schedule $\sched$ from the respective class, and collection of objectives $\varphi_1,\ldots, \varphi_N \in \alpha$, if for every $i \in \set{1,\ldots, N}$ and $\pi_i \in \Conv_\alpha(\G, \sched, \varphi_i)$, we have $(\G,\pol_i,\dots,\pol_N,\sigma) \modelsAS \spec$. 
\item A partial order on conventions: Consider two conventions $\Conv_\alpha$ and $\Conv'_\alpha$. We define $\Conv_\alpha \preceq \Conv'_\alpha$ if $\Conv_\alpha$ uses policies with weaker scheduling information or if $\Conv_\alpha$ is less restrictive than $\Conv'_\alpha$, that is for every $\G, \sched, \varphi$, we have $\Conv'_\alpha(\G, \sched, \varphi) \subseteq \Conv_\alpha(\G, \sched, \varphi)$.
\end{itemize}
\end{definition}

\begin{figure}[ht]
\includegraphics[height=3cm]{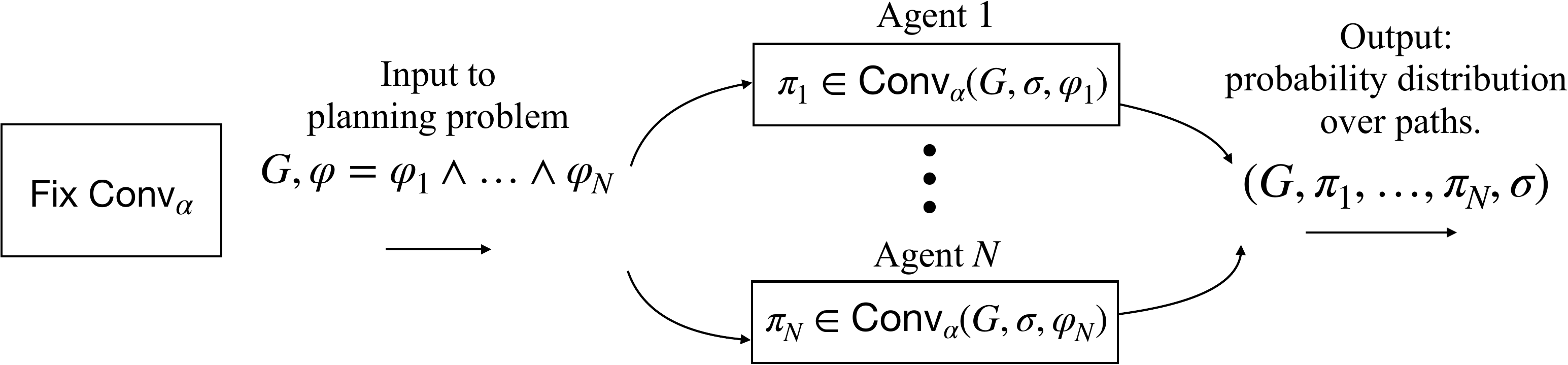}
\caption{Decoupling planning using a convention $\Conv_\alpha$. Observe that $\Conv_\alpha$ is independent of the graph, scheduler, or specific objectives. Agent~$i$ independently chooses $\pi_i \in \Conv_\alpha(\G,\sched,\varphi_i)$ that agrees with the convention. Soundness of $\Conv_\alpha$ implies that the output almost-surely satisfies $\varphi$.}
\label{fig:convention}
\end{figure}

We now present the problem that we will study.
%Guy: This notation appears in the definition.
%We will use the notation $\SAPolFH{N}$, $\SAPolS{N}$, $\SAPolP$ to respectively represent the family of \emph{all} full history-aware, scheduled-aware, and path-aware policies over \emph{all} graphs.

\begin{mybox}[width=\textwidth,colback={blue!10!white},title={Problem statement: decoupled planning of liveness objectives}][
We focus on strongly-connected graphs and fair stochastic schedulers. For a family of objectives $\alpha \in \set{\text{B\"uchi}, \text{co-B\"uchi}, \text{parity}}$. Find a sound and minimal convention $\Conv_\alpha$. 
%That is, for every strongly-connected graph $\G$, stochastic scheduler $\sched$, and objectives $\varphi_1,\ldots, \varphi_N$, we have $(\G,\pol_i,\dots,\pol_N,\sigma) \modelsAS \spec$. 
\stam{
Fix a family of objectives, namely B\"uchi, co-B\"uchi, or parity, and fix a number of objectives $N$.
Determine a family of \emph{sufficiently restricted} set of policies $\SAP\subseteq\SAPolFH{N}\cup \SAPolS{N}\cup \SAPolP$, with $\SAP\neq \emptyset$, such that: 
for every graph $\G$, for every set of liveness objectives $\set{\spec_1,\ldots,\spec_N}$ in $\G$ from the selected family, with $\cap_{i\in [N]}\spec_i \neq \emptyset$, 
every agent $j\in [N]$ can independently select a policy $\pol_j\in \SAP$ that depends on $\G$ and $\spec_j$, so that for every fair scheduler $\sched$, we obtain $\implementation = (\G,\pol_i,\dots,\pol_N,\sigma)\modelsAS \bigcap_{j=1}^N \spec_j$.
}
\end{mybox}

Our focus is on finding sound conventions and we do not show optimality, i.e., provide lower bounds to show inexistence of better conventions. We argue that finding conventions that suffice for decoupling is a nontrivial and challenging problem since conventions are fixed before the graph, scheduler, and objectives are known. Moreover, developing conventions has practical applications whereas showing their inexistence does not.

%We proceed as follows.
%We first establish that we do need restrictions on policies (Sec.~\ref{sec:nonexist}), and then develop restrictions for the prefix-independent case of all objective families (Sec.~\ref{sec:buchi}, \ref{sec:coBuchi}, \ref{sec:parity}).
%Then, we consider the prefix-dependent case, where we slightly modify the construction of our composition to deal with mutual safety violations (prefix-dependent property) by the policies (Sec.~\ref{sec:safety}).

\section{Convention for Decoupling B\"uchi Objectives}\label{sec:buchi}

We show that decoupling is possible for B\"uchi objectives with mildly restricted path-aware policies; that is, B\"uchi admits a strong form of decoupling.

%OLD that if the policies are such that a reachability or B\"uchi set is seen after a finite prefix, if the scheduler chooses each module independently at each step with a fixed positive probability (for simplicity: i.i.d. uniform choice between the two modules), and each policy satisfies a \emph{uniform bounded hitting-time property} defined below, then with probability 1 the resulting random execution satisfies the conjunction of the B\"uchi objectives. 

\begin{definition}[Policies with uniform bounded hitting-time]
Let $\G=(V,E,v_0)$ be a graph and let $\pi\in \Pol^P$ be a stochastic policy on~$\G$.
Fix a subset $B \subseteq V$ and a finite path $\alpha = v_0 v_1 \ldots v_k \in \pathsfin(\G)$.
The \emph{hitting time random variable} of~$B$ under~$\pi$ from~$\alpha$
is defined by
\[
  \tau_\pi(\alpha)(\rho)
  \;:=\;
  \min\{\, t\ge 1 \mid v_{k+t}\in B \text{ in the extension } \rho \in \mathrm{Ext}_\pi(\alpha)\,\},
\]
with the convention $\tau_\pi(\alpha)(\rho)=\infty$ if no such~$t$ exists.
The \emph{expected hitting time} of~$B$ from~$\alpha$ under~$\pi$ is then
$  \mathbb{E}_\pi[\tau_\pi(\alpha)]
  \;=\;
  \int_{\mathrm{Ext}_\pi(\alpha)} \tau_\pi(\alpha)(\rho)\, d\Pr^\pi(\rho)$.
We say that $\pi$ has \emph{uniformly bounded expected hitting time} for~$B$
if there exists a finite constant~$L$ such that  $\sup_{\alpha \in \pathsfin(\G)} 
  \mathbb{E}_\pi[\tau_\pi(\alpha)] \;\le\; L$.
\stam{%short
\[
  \sup_{\alpha \in \pathsfin(\G)} 
  \mathbb{E}_\pi[\tau_\pi(\alpha)] \;\le\; L.
\]
}
% For a finite prefix $\alpha=v_0\ldots v_k$, we define the \emph{hitting time} of a subset of vertices $B$ under $\pi$ from $\alpha$ by
% \[
%   \tau_\pi(\alpha) \;:=\; \min\{\, t\ge 1 \mid v_{k+t}\in B
%   \text{ in } \mathrm{Ext}_\pi(\alpha)\,\},
% \]
% with the convention $\tau_\pi(\alpha)=\infty$ if no such $t$ exists.
% The \emph{hitting time} for vertex set $B$ by a policy $\pi$ is smallest integer $L\ge 1$ such that
% \[
%   \sup_{\alpha\in\pathsfin(\G)} \tau_\pi(\alpha) \;\le\; L.
% \]
% A policy has \emph{uniform bounded hitting time} for $B$ if $L<\infty$. 
\end{definition}

The following theorem proves implementation of any such policies with a stochastic scheduler is good for all. The proof can be found in Appendix~\ref{appn: buchi}. 

%states that any stochastic scheduler that schedules all properties infinitely often with a constant probability is a universal scheduler. 
\begin{restatable}%[Reachability and B\"uchi specifications admit universal schedulers with dense policies with a uniform bound on expected hitting-time]
{theorem}{denseUniformBuchi}\label{thm:denseUniformBuchi}
Consider a strongly connected graph $\G$, policies $\pol_1,\ldots, \pol_N$ with uniform bounded hitting-time for B\"uchi objectives $\spec_1,\dots,\spec_N$, and a stochastic scheduler $\sched$. Then, $(\G, \pol_1,\ldots, \pol_N,\sched) \modelsAS \bigcap_{1 \leq i \leq N} \spec_i$.
%Under fair schedulers, the sufficient condition required for B\"uchi protocol to be strong and 
\end{restatable}

\stam{%old
\begin{corollary}\label{cor:Buchi+dense+bounded hitting+uniform scheduler}
    The uniform scheduler $\sigma:[N]^\star\mapsto [N]$, that is,  $\Pr[j = \sigma(u)]=1/N$ for all $u\in [N]^\star$, and $\forall j\in [N]$ is a universal scheduler for B\"uchi objectives on strongly connected graphs. 
\end{corollary}
\begin{corollary}
    Any scheduler $\sigma:[N]^\star\mapsto [N]$ such that there is a positive constant $\epsilon$ such that $\Pr[j = \sigma(u)]>\epsilon$ for all $u\in [N]^\star$, and $\forall j\in [N]$ is a universal scheduler for B\"uchi objectives on strongly connected graphs for finite-memory policies. 
\end{corollary}
}
\section{Convention for Decoupling Co-B\"uchi Objectives}\label{sec:coBuchi}
%\subsection{On the non-existence of scheduler-unaware local policies}
We start by showing that restricting to dense path aware policies does not suffice for co-B\"uchi objectives. 
%In the situation in the case of coB\"uchi, we not only show that we cannot compose two arbitrary policies, even when they are finite memory using fair schedulers. universal schedulers, we show that there are no universal schedulers when the agents are restricted to finite-memory or even memoryless. 
\begin{restatable}%[Dense path aware policies for coB\"uchi specifications are not enough]
{theorem}{NoUnivCoBuchi}\label{theorem:NoUnivCoBuchi}
There is a graph, two 
%For $N=2$, and any scheduler $\sched$, there exists 
co-B\"uchi objectives $\spec_1$ and $\spec_2$ with $(\spec_1\cap\spec_2) \neq \emptyset$ and deterministic memoryless policies $\pol^\sched_1$ and $\pol^\sched_2$ that respectively satisfy $\spec_1$ and $\spec_2$, but for any scheduler $\sched$, we have $(\G,\pol_1^\sched,\pol_2^\sched,\sched) \not \modelsAS \spec_1\cap\spec_2$.     
\end{restatable}
\begin{proof}[sketch]
Consider the graph $\G$ depicted in \cref{figure:recoBuchiNotUniversal}. For $i \in \set{1,2}$, the ``bad'' co-B\"uchi vertex for objective $\varphi_i$ is $v_{1i}$. Note that $(\spec_1\cap\spec_2) \neq \emptyset$ since $v_0^\omega$ satisfies both objectives. We define $\pi_i(v_1) = v_{1,i}$, otherwise the policies are identical. Observe that $\pi_i$ satisfies $\varphi_i$. However, for any stochastic scheduler, all vertices will be almost-surely visited, which violates both objectives. 
The proof can be found in App.~\ref{app:NoUnivCoBuchi}. 
\stam{%OLD
We construct an explicit counterexample showing that no scheduler can compose two path-aware policies for coB\"uchi policies. 
The graph $\G$ (see \cref{figure:recoBuchiNotUniversal}) is such that $v_{11}$ and $v_{12}$ are the respective coB\"uchi vertices for the specifications $\spec_1$ and $\spec_2$. 
Both policies $\pi_1$ and $\pi_2$ behave identically everywhere except at $v_1$, where $\pi_1$ moves to $v_{12}$ and $\pi_2$ moves to $v_{11}$. 
Thus, each policy individually satisfies its own specification.%, and there exists a path (e.g.\ looping at $v_0$) that avoids both bad vertices, so $\spec_1\cap\spec_2\neq\emptyset$.

Now consider any scheduler $\sigma$. 
If $\sigma$ eventually always schedules one agent (say agent~1), then the composition behaves as $\pi_1$ from some point onward, and visits $v_{12}$ infinitely often, violating $\spec_2$. 
Hence, $\spec_1\cap\spec_2$ is violated on a set of positive probability. 
If, on the other hand, $\sigma$ schedules both modules infinitely often, then both $v_{11}$ and $v_{12}$ are visited infinitely often, violating both coB\"uchi objectives. 
Therefore, in every case, the composition $(\G,\pi_1,\pi_2,\sigma)$ fails to almost-surely satisfy $\spec_1\cap\spec_2$, establishing that no scheduler can be universal for coB\"uchi objectives.
}
\end{proof}

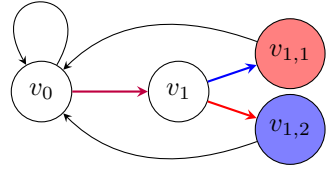
\begin{wrapfigure}{r}{5cm}
\centering
\begin{tikzpicture}[->,>=stealth,node distance=1cm]
  \tikzstyle{normal}=[circle,draw,minimum size=8mm]
  \tikzstyle{bad1}=[circle,draw,fill=red!50,minimum size=8mm]
  \tikzstyle{bad2}=[circle,draw,fill=blue!50,minimum size=8mm]
  \tikzstyle{shared}=[circle,draw,fill=purple!40,minimum size=8mm]

  \node[normal] (v0) {$v_0$};
  \node[normal, right=1cm of v0] (v1) {$v_1$};
  % \node[normal, left=1.5cm of v0] (v2) {$v_2$};

  \node[bad1, right=0.6cm of v1, yshift=+0.5cm] (v11) {$v_{1,1}$};
  \node[bad2, right=0.6cm of v1, yshift=-0.5cm] (v12) {$v_{1,2}$};

  % \node[shared, left=1cm of v2, yshift=+1cm] (v21) {$v_{21}$};
  % \node[normal, left=1cm of v2, yshift=-1cm] (v22) {$v_{22}$};
%=======
%  \node[bad1, right=0.6cm of v1, yshift=+0.5cm] (v11) {$v_{11}$};
%  \node[bad2, right=0.6cm of v1, yshift=-0.5cm] (v12) {$v_{12}$};
%>>>>>>> refs/remotes/origin/main

  \draw[purple,thick] (v0) -- (v1);
  % \draw (v0) -- (v2);

  \draw[red,thick] (v1) -- (v12);   
  \draw[blue,thick] (v1) -- (v11);  

  % \draw (v2) -- (v21);
  % \draw (v2) -- (v22);

  \draw (v11) to [out=160,in=45] (v0);
  \draw (v12) to [out=200,in=-45] (v0);
\draw (v0) to [out=60,in=120,looseness=8] (v0);

\end{tikzpicture}
\caption{A graph with two co-B\"uchi objectives given by bad states $v_{1,1}$ ({\color{red} Red}) and $v_{1,2}$ ({\color{blue} Blue}). Red edges are chosen by $\pol_1$ and Blue edges by $\pol_2$.}
\label{figure:recoBuchiNotUniversal}
\end{wrapfigure}

%
%\section{Deterministic Universal Schedulers?}
%\input{3Deterministicschedulers}
%\section{Stochastic universal schedulers}
%\input{4stochasticScheduler}
%\section{Universal Local policies}
\subsection{A convention for co-B\"uchi objectives}

%We proceed to describe decoupling for co-B\"uchi objectives using scheduled-aware policies. We describe a convention for how an agent chooses a policy. We stress that the convention is fixed {\em before} knowing the specific graph and co-B\"uchi objective. Thus, the convention can be applied by the agents independently. 
	
Consider a strongly-connected graph $\G$, a stochastic scheduler $\sched$, and a co-B\"uchi objective $\varphi$ given by a set of ``bad'' states $B$. Let $\mathfrak{C}_B$ denote the set of ``good'' simple cycles in $\G$ that are disjoint from $B$. A policy $\pol \in \Conv_{\text{co-B\"uchi}}(\G, \sched, \varphi)$ proceeds as follows. It randomly chooses a lasso path $\rho \cdot C^\omega$, where $\rho$ is a simple path and $C \in \mathfrak{C}_B$. 
If $\pol$ is scheduled in a turn, if $C$ has not yet been reached, it chooses the next vertex on $\rho$, and otherwise advances on $C$. Suppose that $\pol$ is not scheduled, and let $\pol'$ be the scheduled policy. Let $v$ be the vertex that $\pol$ would have chosen if it was scheduled and let $v'$ be the vertex that $\pol'$ chooses. If $v = v'$, we proceed to the next turn. If $v \neq v'$, then $\pol'$ realizes it is in {\em conflict} with $\pol$: the lasso that $\pol$ chose differs from the choice of $\pol'$. Then, $\pol'$ randomly chooses a different lasso path $\rho' \cdot C'^\omega$ with $C' \in \mathfrak{C}_B$. Correctness will be derived from the property of co-B\"uchi objectives that if there is a path that satisfies all objectives, then there is a lasso $\rho \cdot C^\omega$ with $C \in \mathfrak{C}_B$ that is good for all.  The idea is that almost surely, all policies will eventually choose the same lasso path leading to the generated path ``stabilizing'' on a cycle that is good for all.

We describe $\pol$ formally. The memory states of $\pol$ are $M = \pathsfin(\G) \times \mathfrak{C}_B$, where a state $\zug{\rho, C} \in M$ means that $\pol$ chooses cycle $C$ and the path $\rho$ leading to it. 
The interaction with a scheduler and the other policies generates a random sequence, and $\pol$'s {\em local} view of the sequence is $\zug{m_0, v_0}, \zug{m_1, v_1},\ldots$, where for $i \geq 0$, each $m_i = \zug{\eta_i, C_i} \in M$ is a memory state and $v_i \in V$ is a location on $\G$. 
%For $j \geq 0$, we denote the $j$-th vertex of $\eta$ by $\eta[j]$ and by $\eta[j:]$ as the suffix starting from the $j$-th vertex. 
The policy maintains the invariant that the location in $\G$ matches the first vertex of $\eta$, namely $\eta_i[0] = v_i$. 
We use $\top$ and $\bot$ to respectively denote that $\pol$ is and is not scheduled. 
The definition of $\zug{m_{i+1}, v_{i+1}} = \pol(m_i, v_i, \top)$ is deterministic and depends on whether the cycle has been reached: (1)~$C_i$ has not been reached: then we define $v_{i+1} = \eta[1]$ and $m_{i+1} = \zug{\eta[1:], C_i}$, and (2)~$C_i$ has been reached: then $\eta_i= \epsilon$ and we define $m_{i+1} = m_i$ and $v_{i+1}$ is the successor of $v_i$ on $C$. 
We proceed to define $\pol(m_i, v_i, \bot)$. Let $\zug{m, v} = \pol(m_i, v_i, \top)$ be $\pol$'s choice if it was scheduled, and let $v_{i+1}$ be the vertex chosen by some other policy. If $v = v_{i+1}$, then the update is as above, $m_{i+1} = m$. 
If $v_{i+1} \neq v$, then $\pol$ randomly selects $C_{i+1} \in \mathfrak{C}_B$ and a simple path $\eta_{i+1}$ from $v_{i+1}$ to $C_{i+1}$.

To prove the following theorem, we construct a Markov chain that simulates the global state of the policies and show that its Bottom Strongly Connected Components (BSCCs) are states in which the policies are in consensus regarding the choice of a lasso path that is good for all. The proof then follows from the property of Markov chains that a random walk reaches a BSCC almost surely. 
The full proof is available in Appendix~\ref{app:coBuchischeduling}.

\begin{restatable}{theorem}{coBuchischeduling}\label{thm: coBuchischeduling}
Let $B_1,\ldots, B_N$ be a collection of co-B\"uchi objectives on a graph $\G$. For every choice of policies that satisfies the co-B\"uchi condition $\pol_i \in \chi_{B_i}$, for $i \in [N]$, and a fair scheduler $\sched$, we have $(\G, \pi_1,\ldots, \pi_N, \sched) \modelsAS \bigwedge_{1 \leq i \leq N} \varphi_i$. 
\end{restatable}

\section{Convention for Decoupling Parity Objectives}\label{sec:parity}

We generalize the ideas for decoupling co-B\"uchi objectives to parity objectives. The construction requires full-history aware policies when $N > 2$. Fix a graph $\G$ and stochastic scheduler $\sched$. 

The co-B\"uchi convention relies on the property that if there is a path that satisfies all objectives, there is such a lasso path. This is no longer the case for parity objectives. For example, consider the graph depicted in \ref{subfig:Leftintro} and objectives $\varphi_1 = (GF \ell) \wedge (G \neg b_1)$ and $\varphi_2 = (GF r) \wedge (G \neg b_1)$, then a path that satisfies both must cycle between $r$ and $\ell$ via $t_1$ and $t_2$. This is not a simple cycle since both outgoing edges of $t_1$ and $t_2$ are traversed infinitely often. We overcome this using the following lemma whose proof can be found in App.~\ref{app:kappagood}.
\begin{lemma}\label{lemma: kappagood}
	Consider parity objectives \(\kappa_1, \ldots, \kappa_N:V \rightarrow \mathbb{N}\) with $\bigwedge_{i\in [1;N]}\kappa_i \neq \emptyset$. Then, there are memoryless policies \(\zug{\theta_1, \ldots, \theta_N}: V \rightarrow V\) such that \((\G, \theta_1, \ldots, \theta_N\), \(\sigma) \modelsAS \bigwedge_{i\in [1;N]}\kappa_i\)
\end{lemma}
For example, \ref{subfig:Leftintro} depicts two memoryless policies $\pi_1$ and $\pi_2$ for which we have $(\G, \pi_1,\pi_2, \sched)\modelsAS \varphi_1 \wedge \varphi_2$ since the generated random walk almost surely visits each of the vertices $\ell, t_1, t_2, r$ infinitely often.

\smallskip
\noindent
\textbf{Parity convention.}
Consider a parity objective given by a colouring function \(\kappa_i: V \rightarrow \mathbb{N}\). 
%, and let \(\mathfrak{C}_\kappa\) denote the set of simple cycles in \(\G\) that has the largest colour even. 
A policy $\pol_i \in \Conv_{\text{Parity}}(\G, \sched, \kappa)$ randomly chooses an \(N\)-tuple of {\em internal} memoryless policies \(\theta_1, \ldots \theta_N\) such that $(\G, \theta_1, \ldots \theta_N, \sched) \modelsAS \kappa_i$. 
%To reduce confusion, we call $\theta_i$ an {\em internal policy}. 
Intuitively, when scheduled, $\pol_i$ follows $\theta_i$, and $\theta_j$, for $j \neq i$, constitutes a guess of the internal policy that $\pol_j$ follows. 
At vertex $v$, if $\pol_i$ is scheduled, it proceed to $\theta_i(v)$ and if $\pol_j$ is scheduled and chooses $u \neq \theta_i(v)$, then $\pol_i$ observes a conflict and randomly chooses a new $N$-tuple of internal memoryless policies. 
In order to observe a conflict, $\pol_i$ requires knowledge of who is scheduled and their choice, thus it is a full-history aware policy.

We describe \(\pol_i\) formally. 
Its memory states are \(M = (V^V)^N\). 
 The interaction with a scheduler and the other policies generates a random sequence, and \(\pol_i\)'s \emph{local} view of the sequence is \(\zug{m_0, v_0}, \zug{m_1, v_1}, \ldots\), where at time \(j \geq 0\), the location in $\G$ is \(v_j \in V\) and the memory state is \(m_j = \zug{\thetatwo{j}{1}, \ldots \thetatwo{j}{N}} \in M\). 
 If $\pol_i$ is scheduled at time $j$, we define $v_{j+1} = \thetatwo{j}{i}(v_j)$ and the memory state is unchanged, namely $m_{j+1} = m_j$. 
 Suppose that $\pol_{i'}$ is scheduled and chooses $u$, then we define $u = v_{j+1}$. If $u = \thetatwo{j}{i}(v_j)$, then there is no conflict and $m_{j+1} = m_j$. 
 Otherwise, $u \neq \thetatwo{j}{i}(v_j)$, policy $\pol_i$ realizes a conflict and randomly chooses \(m_{j+1} = \zug{\thetatwo{j+1}{1}, \ldots \thetatwo{j+1}{N}} \in M\).

\smallskip
The following theorem is obtained by reasoning about the Markov chain that simulates the global state of the policies, carefully identifying consensus global states that are trivially good for all, and showing that they form the BSSCs of the Markov chain. 
The proof is available in Appendix~\ref{app:parityscheduling}. 

\begin{restatable}{theorem}{parityscheduling}\label{thm: parityscheduling}
	Consider a strongly-connected graph $\G$, a stochastic scheduler $\sched$, and a collection of $N$ parity objectives  $\kappa_1,\ldots, \kappa_N$ having $\bigcap_{1 \leq i \leq N} \kappa_i \neq \emptyset$.
	For every collection of policies $\pi_1,\ldots, \pi_N$ that satisfy the convention, i.e., $\pi_i \in \Conv_{\text{parity}}(\G, \sched, \kappa_i)$, for $i \in \set{1,\ldots, N}$, we have $(\G, \pi_1,\ldots, \pi_N, \sched) \modelsAS \bigwedge_{1 \leq i \leq N} \kappa_i$. 
\end{restatable}

\section{Discussion}
% We conclude by summarising our main technical results across different prefix-dependent families of objectives.

% \begin{description}
%   \item[Reachability and B\"uchi objectives.]
%   For both reachability and B\"uchi objectives on strongly connected graphs, we identify a broad class of \emph{path-aware} (regular) strategies that admit universal schedulers. 
%   We show that these are policies that visit the designated B\"uchi vertices infinitely often and, moreover, have a uniformly bounded expected hitting time after every prefix. 
%   Finite-memory strategies that fulfil their respective B\"uchi objectives after any prefix are one such class satisfying this property.

%   \item[Co-B\"uchi objectives.]
%   For co-B\"uchi objectives, we show that \emph{scheduler-aware} policies---constructed so that their memory states correspond to simple cycles (lassos) in the graph---can be composed under any fair scheduler. 
%   This class captures the minimal scheduler information necessary for almost-sure satisfaction.

%   \item[Parity objectives.]
%   For general parity objectives, we establish that \emph{full-history aware} policies are required. 
%   In this setting, each policy’s memory corresponds to an $N$-tuple of local strategies (one for each module), as discussed in \cref{sec:parity}. 
%   Such policies remain composable under any fair scheduler.

%   \item[Safety shields.]
%   Finally, we discuss how safety shields can be implemented for compositions of policies that are mutually safety-closed, ensuring that the composed execution never leaves the safe region.
% \end{description}

We established sufficient conditions under which local policies for different $\omega$-regular objectives can be composed under a common scheduler to almost-surely satisfy all objectives, showing that B\"uchi, co-B\"uchi, and parity objectives require progressively richer scheduler awareness.

Although our focus is on multi-parity objectives, the framework and the existence results for universal local policies naturally extend to any class of specifications implementable via finite-memory strategies, provided the memory bound is common knowledge among all modules. 
In this setting, each module may guess a finite-memory strategy per agent (rather than a memoryless one) whose bounded memory size is shared across modules.

Interestingly, while our analysis in \cref{sec:parity} shows that full-history awareness is necessary in general, there are special cases that require less information. 
For instance, when there are only two agents with parity objectives, each can infer the scheduler’s choice from the observed transition: if the transition corresponds to its own proposed action, the agent was scheduled; otherwise, the other agent was. 
Hence, no explicit communication from the scheduler is needed in this case.

For parity objectives, the memory size of each agent typically scales with the number of agents, as each agent must maintain a guess of a cycle for every other agent (including itself). 
However, if an agent has additional knowledge—e.g., that some others have B\"uchi or co-B\"uchi objectives—this requirement can be reduced. 
When a parity objective is combined with a (co-)B\"uchi objective, satisfiability can be achieved via a single simple cycle, eliminating the need for additional memory components.

Although we establish the existence of such universal schedulers, our work does not address their convergence properties. In particular, the expected convergence time of the presented protocol is likely to be exponential in the size of the underlying system. Improving this convergence time remains is not discussed in this paper. We conjecture that techniques from reinforcement learning or stochastic approximation could be adapted to accelerate convergence while preserving correctness guarantees, and is left for future work.

 \bibliographystyle{splncs04}
 \bibliography{refs}
\appendix
%\section{Appendix for~\cref{sec:nonexist}}\label{appn: A}
\section{Appendix for Section~\ref{sec:framework}}\label{appn: A}
We first state well-known results from measure theory that we use in our proofs.
\begin{theorem}[Carathéodory Extension Theorem \cite{Bil95}] 
Let $\Omega$ be a set and $\mathcal{F}_0$ be a field of subsets of $\Omega$. Let $P$ be a probability measure on $\mathcal{F}_0$. There exists a unique probability measure $Q$ on $\sigma(\mathcal{F}_0)$ such that $Q(A) = P(A)$ for all $A \in \mathcal{F}_0$.
\end{theorem}

\begin{lemma}[Borel-Cantelli lemmas \cite{Bor09,Can17,Shi16}]\label{lemma:borellcantelli}
    Let $E_1,E_2\dots,$ be an infinite sequence of events in some probability space.
    \begin{description}
    	\item[First lemma:] If $\sum_{n=1}^\infty\Pr[E_n]<\infty$, then the probability that infinitely many events occur is $0$, namely $\Pr[\limsup_{n\to\infty}E_n] = 0$.
    	\item[Second lemma:] If the events are all \emph{independent} and $\sum_{n=1}^\infty\Pr[E_n] = \infty$, then almost-surely infinitely many of them occur, namely $\Pr[\limsup_{n\to\infty}E_n] = 1$.
    \end{description}  
\end{lemma}

The Borel-Cantelli lemmas can be intuitively explained using a coin toss example.
Suppose we are tossing a coin whose bias is varying over time.
Let $p_i$ be the probability of seeing heads in the $i$-th toss.
We are interested to know: if we toss the coin infinitely many times, will we see infinitely many heads?
The first lemma says that if the infinite sum $\sum_{i=1}^\infty p_i$ converges to a finite number, then the probability of seeing infinitely many heads is zero; an example would be $p_i = 2^{-i}$.
The second lemma says that if the infinite sum $\sum_{i=1}^\infty p_i$ diverges to infinity, then the probability of seeing infinitely many heads is one; an example of this case would be $p_i$ being a constant.

\subsection{Proof of Thm.~\ref{prop:NoDetScheduler}}
\label{app:NoDetScheduler}

\NoDetSched*

\begin{proof}
	For the contrary, we suppose \(\sched\) is a universal deterministic scheduler. 
	In the following, we exhaustively explore all the possibilities of how \(\sched\) may behave. 
	 Accordingly, we provide a game graph \(\mathcal{G}\), and for \(N = 2\), a pair of objectives and policies which will violate the universality condition of \(\sched\). 
	 
	 We consider the graph \(G\) as shown in Fig~\ref{fig: determinB}, where the set of vertices are \(V = \{v, a_1, a_2\}\), and the set of edges are \(E = \{(v, a_1), (v, a_2), (a_1, v), (a_2, v), (v, v)\}\). 
	 The initial vertex is \(v\).  
	 Moreover, module 1 and 2's target set of vertices is the singleton set consisting of \(a_1\) and \(a_2\), respectively. 
	 We describe the two exhaustive cases in the following:
	 
%	 (a) After a finite prefix, \(\sched\) only schedules one of the module forever. 
%	  Suppose Policy~\(i\) chooses to go to \(a_i\) whenever it is scheduled and the game is at \(v\). 
%	 In this case, if \(\sched\) only schedules module \(i\) (for \(i \in \{1, 2\}\)) after a finite prefix, then \(\spec_{3-i}\) would not be fulfiled. 

%	 (b) \(\sched\) schedules both the modules infinitely often.
	% In this case, we do not explicitly construct the two policies, rather implictly 
    We construct 
    %show the existence of 
    two policies for which \(\sched\) fails to maintain the universality condition.  
	 % Since, \(\sched\) is a deterministic scheduler, there is a fixed infinite sequence \(p_1 p_2 \ldots \in \{1, 2\}^\omega\) that it produces. 
     Let the infinite sequence $p_1p_2\dots \in \{1,2\}^\omega$ be the unique path produced by the deterministic scheduler $\sched$.
     The only outgoing edge from vertex $a_1$ as well as from $a_2$ is the central vertex $v$. 
%	 It is clear from the graph that when the game is at vertex \(a_1\) or \(a_2\), none of the module has any choice but to choose the edge that makes the next vertex \(v\).  
	 %It is possible that there exists a 
     We define policy \(\pol_1\) which behaves in the following manner: at the $j^\textsuperscript{th}$ iteration, that is, after $j$ time-steps have elapsed (assuming the game is at vertex \(v\)), if \(p_j = 1\), then \(\pol_1\) prescribes taking the self-loop at \(v\), and if \(p_j = 2\), then \(\pol_1\) prescribes the edge \((v, a_1)\).
	 Similarly, we define \(\pol_2\) as above replacing the role of $1$ and $2$. 
	 This makes the game stay at vertex \(v\) forever for the given scheduler $\sigma$. 
	 Thus, the implementation \((\G,\pol_1,\pol_2,\sched) \not\models \spec_1 \cap \spec_2\), even though \(\spec_1 \cap \spec_2 \neq \emptyset\) and \(\pol_i \models \spec_i\) for \(i = 1, 2\). 
    \qed
\end{proof}

\subsection{Proof of Thm.~\ref{thm:no univ buchi scheduler}}
\label{app:no univ buchi scheduler}

\NoUnivBuchischeduler*

\begin{proof}

    Consider a fixed scheduler $\sigma$. 
    
    We consider the following graph $\G$ in \cref{fig:screenshotTempBuchi}, which consists of vertices $V= \{\,v_0,\; a_1,\; a_2,\; b_1,\; b_2\,\}$, and edges $E = \{\,(v,a_i),\; (a_1,a_1),(a_1,v_0),(a_1,b_1),\; (b_i,v)\mid i\in\{1,2\}\,\}$. We consider $v_0$ to be the initial vertex.

    We define the specifications $\spec_1$ and $\spec_2$ for modules $1$ and $2$ respectively by their B\"uchi sets that need to be visited infinitely often, where $\spec_1$ requires that the vertex $b_{1}$ is seen infinitely often, and $\spec_2$ requires that $b_2$ is seen infinitely often. 
    
    The policies that we consider in the proof have the following shape. 
    For module $i \in \{1,2\}$ that, when vertex $v$ is visited for the $n^\textsuperscript{th}$ time after which policy $i$ is selected, the policy $\pi_i$ choses to go to vertex $a_i$, and  the policy $\pi_i$ then counts till a large value $L(n)$ many steps and stays at vertex $a_i$ using the self-loop for $L(n)$ many steps. After this $L(n)$ many steps, then it proposes to visit vertex $b_i$. 
    This $L(n)$ is an fast-increasing sequence constructed based on the scheduler $\sigma$. 
%Guy: we don't define densely fulfils anymore
%    For any sequence $L$ from $\Nat\to \Nat$, note that such policies $\pol_i$ densely fulfils $\spec_i$ for $i\in\{1,2\}$. 

    Since we have a fixed scheduler $\sigma$, we define the following events $\mathrm{Lock}_1$ and $\mathrm{Lock}_2$
    \[
  \mathrm{Lock}_j \;=\; \{\,s\in\{1,2\}^\omega \mid \exists t\ \forall r\ge t:\ s_r=j\,\}
\] 
Since $\sigma$ induces a probability distribution $\Pr^\sigma$ on $\{1,2\}$, the values $\Pr^\sigma(\mathrm{Lock}_j)$ is well-defined for $j\in\{1,2\}$.

    We call each visit to $v$ an ``attempt'', since process $i$, attempts to visits its B\"uchi state. 
        For each ``attempt'' $n\ge 1$ and module $j\in\{1,\dots,N\}$ define
\[
  p_{n,j}(L)=\Pr_\sigma\bigl(\text{indices at the $n$-th attempt
    start forming a block of $L$ consecutive $j$'s}\bigr),
\]
\[
  p_{n,j}(\infty)=\lim_{L\to\infty} p_{n,j}(L)
  = \Pr_\sigma\bigl(\text{all indices from the $n$-th attempt start
    onward are } j\bigr).
\]
For a given scheduler~$\sigma$ one of the following is true: either~(a)~there exists a module $j$ such that $\Pr_\sigma(\mathrm{Lock}_j)>0$; or (b) $\Pr_\sigma(\mathrm{Lock}_1)=\Pr_\sigma(\mathrm{Lock}_2)=0$.

\paragraph{(a) One module is chosen forever after a point with positive probability}. There exists $j$ with $\Pr_\sigma(\mathrm{Lock}_j)>0$.
In this case, we can argue that the B\"uchi condition of the module that is not $j$ is therefore cannot be achieved with at least this positive probability. 

\paragraph{(b) No module is chosen forever after a point with positive probability, the probability of seeing either module ``forever'' is $0$.} In this case, for both modules, that is modules $j\in \{1,2\}$, we have $\Pr_\sigma(\mathrm{Lock}_j)=0$; moreover,
        for each fixed time $t\in\mathbb{N}$ and module $j \in \{1,2\}$, there are functions $L_1,L_2\,:\,\mathbb{N}\to\mathbb{N}$ such that
        \(          \lim_{L_j\to\infty}\Pr_\sigma(s_t=s_{t+1}=\cdots=s_{t+L_j-1}=j)=0.
        \)
        Therefore, we show in the claim, that one can choose an increasing sequence $L_j(n)\to\infty$ such that, for every finite collection of absolute starting times $(t_n)$ corresponding to the intended attempt-starts, the probabilities of an $L_j(n)$-long block of $j$ beginning at $t_n$ are such that their summation is finite. 
        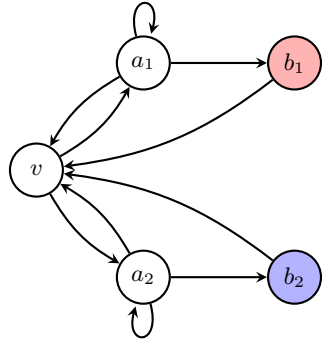
\begin{figure}
            \centering
            \begin{tikzpicture}[>=stealth,thick,node distance=2cm]

\tikzstyle{state}=[circle,draw,minimum size=7mm,inner sep=0pt]
\tikzstyle{target1}=[circle,draw,minimum size=7mm,inner sep=0pt,fill=red!30]
\tikzstyle{target2}=[circle,draw,minimum size=7mm,inner sep=0pt,fill=blue!30]

% Nodes
\node[state] (v) {$v$};
\node[state,above right of=v] (a1) {$a_1$};
\node[state,below right of=v] (a2) {$a_2$};
\node[target1,right of=a1] (b1) {$b_1$};
\node[target2,right of=a2] (b2) {$b_2$};

% Edges
\draw[->] (v) edge[bend right=15] (a1);
\draw[->] (v) edge[bend right=15] (a2);

\draw[->] (a1) edge[loop above] ()
          (a1) -- (b1)
          (a1) edge[bend right=15] (v);

\draw[->] (a2) edge[loop below] ()
          (a2) -- (b2)
          (a2) edge[bend right=15] (v);

\draw[->] (b1) edge[bend left=15] (v);
\draw[->] (b2) edge[bend right=15] (v);

\end{tikzpicture}
            \caption{This figure is redrawn for convenience. 
            %The figure in \cref{fig:TempBuchi} redrawn for convenience. 
            Red is module 1's B\"uchi state, and Blue is module 2's B\"uchi state}
\label{fig:screenshotTempBuchi}
        \end{figure}

Using our claim, and setting $L(n) = \max\{L_1(n),L_2(n)\}$, we can construct policies having the shape discussed in the beginning of the proof.  Recall the strategy is that when v is visited for the $n^\textsuperscript{th}$ time stays at $a_i$ for $L(n)$ many steps. Then it visits $b_i$ and returns. But the process $j\neq i$ always wants to leave vertex $a_i$ back to $v$. 
The probability of visiting both $b_1$ and $b_2$ infinitely often then is $0$, due to Borel-Cantelli lemma~(\cref{lemma:borellcantelli}).

%         The rough idea is to build a policy for module $i \in \{1,2\}$ that, when v is visited for the $n^\textsuperscript{th}$ time stays at $a_i$ for some large $L_i(n)$ many steps. Then it visits $b_i$ and returns. But the process $j\neq i$ always wants to leave vertex $a_i$ back to $v$. 
%      The trick is to construct growing functions $L_i$ based on $\sigma$.

% \[
%   \mathrm{Lock}_j \;=\; \{\,s\in\{1,2\}^\omega \mid \exists t\ \forall r\ge t:\ s_r=j\,\}
% \]

\begin{claim}\label{claim:LnExists}
Assume that for all $j\in\{1,2\}$, $\Pr(\mathrm{Lock_j} = 0)$. Then for each module $j$, there is an increasing sequence $L(n)\to\infty$ (for a fixed scheduler~$\sigma$) such that, for every finite collection of $(t_n)$ corresponding to the intended attempt-starts, the probabilities of an $L(n)$-long block of scheduling the module $j$ beginning at time $t_n$ are such that their summation is finite.
\end{claim}
        By our assumption this limit is $0$ for every fixed $t$. Thus for each fixed $t$ we have $\lim_{L\to \infty}p_{t,j}(L)= 0$.

        We now construct $L(n)$ inductively so that $\lim_{n\to\infty} L(n) = \infty$ and $p_{t_n,j}(L(n))$ is small enough that the infinite sum is finite. Set \(L(0)=0\).
        For \(n\geq 1\) choose \(L(n)>L(n-1)\) large enough that
\[
  p_{t_n,j}\bigl(L(n)\bigr) \le 2^{-n-1}.
\]
This is possible because \(p_{t_n,j}(L)\to 0\) as \(L\to\infty\). The sequence \(L(n)\) so obtained tends to infinity and satisfies
\[
  \sum_{n=1}^\infty p_{t_n,j}\bigl(L(n)\bigr)
  \leq
  \sum_{n=1}^\infty 2^{-n-1}
  = 1/2,
\]
as required. 
% \[
        %   \sum_{n=1}^\infty p_{n,j}\bigl(L(n)\bigr) < \infty
        %   \qquad\text{for } j=1,2.
        % \]
\qed
\end{proof}

\section{Appendix for~\cref{sec:safety}}\label{appn: safety}
Let $W_i$ be the maximal set of vertices in $\G$ from which the agent $i$ has a policy to fulfil $\spec_i$.
	In other words, $\specsafe_i = \safe(W_i)$.
	Define $W\coloneqq \cap_{i\in [N]} W_i$. 
	By assumption, $W$ is nonempty, since $\vinit\in W_i$ for all $i$.	
	Let $\G[W]$ be the subgraph obtained by removing all the vertices $\V\setminus W$ from $\G$, and by removing all the edges that involve these vertices.

\begin{proposition}\label{Prop:subgraph is dead end free}
	The subgraph $\G[W]$ is deadlock-free, meaning every vertex in $W$ has a successor in $W$.	
\end{proposition}

\begin{proof}
	This follows from the mutual safety-closure:
	imagine, towards contradiction, that there is a vertex $v\in W$ that has no successors.
	It follows that for every successor $y$ of $v$ in $\G$, there exists a $i\in [N]$ such that $y\notin W_i$.
	But this contradicts the mutual safety closure assumption, because every path $\path$ in $\G$ ending at $y$ is in $\cap_{i\in [N]} W_i=\cap_{i\in [N]} \pref(\spec_i)$, but not in $\pref\left(\cap_{i\in [N]} \spec_i\right)$.
	We conclude that the subgraph $\G[W]$ is deadlock-free.
\end{proof}

\safetyclosed*

\begin{proof}
	First observe that $\cap_{i=1}^N\speclive_i \neq \emptyset$, which follows from the assumption that $\cap_{i=1}^n\spec_i = \emptyset$, and that $\spec_i = \specsafe_i\cap\speclive_i$ for every $i$.
	For each $i\in [N]$, we write $\speclive_i[W]$ to be the mapping of $\speclive_i$ as an objective in $\G[W]$, defined as: $\speclive_i[W] \coloneqq \speclive_i \cap \pathsinf(\G[W])$.
	We argue that $\speclive_1[W]\cap \ldots \cap \speclive_N[W] \neq \emptyset$.
	This follows from the assumption that $\spec_1\cap\ldots\cap \spec_N\neq \emptyset$ in $\G$, which implies that there is an infinite path $\path$ in $\G$ that is in $\specsafe_i$ and $\speclive_i$ for every $i \in [N]$; in other words, $\path$ is in $\cap_{i\in [N]}\specsafe_i$ as well as in $\cap_{i\in [N]}\speclive_i$.	
	We first argue that $\path$ must be in $W^\omega$.
	If this were not the case, then there would be a largest prefix of $\path$ in $W^*$ (it exists, because $\vinit\in W$), before $\path$ exits $W$.
	However, since $W=\cap_{i\in [N]} W_i$, and since, by Prop.~\ref{Prop:subgraph is dead end free}, no path in $\G[W]$ ends in a dead end, this would imply that $\path$ would exit also from one of the $W_i$-s, which would violate $\specsafe_i$.
	Therefore, $\path \in W^\omega$.
	Now since $\path$ is also in $\cap_{i\in [N]}\speclive_i$, it trivially follows that $\path \in W^\omega \cap \left[ \cap_{i\in [N]}\speclive_i \right]\cap \left[ \cap_{i\in [N]}\specsafe_i \right]= \cap_{i\in [N]} \speclive_i[W]$.

	Now we show that $\implementation'=(\G,\pol_1,\dots,\pol_N,\sigma)^\sh \modelsAS \safe(W)$. 
	Suppose, for contradiction's sake, $\Pr^{\implementation'}(\safe(W)) < 1$.
	This implies that there is a path $\path v\in \pathsfin \cap W^*$ in $\G$ such that $\Pr^{\implementation'}(\path v) > 0$, and there exists an agent $i$ who selects $v'$ as a successor of $v$ with positive probability such that $v'\notin W$.
	This means, there exists a $j\in [N]$ such that $v'\notin W_j$.
	However, by construction of shielded implementations, this is not possible, since the $j$-th agent would ``block'' $v'$ by excluding it from the safe successors sent to agent $i$, and agent $i$ must not choose $v'$ in this case: a contradiction.
	
	The above result allows us to ignore the vertices outside of $W$ and focus on $\G[W]$.
	Clearly, if the policy $\pol_i$ for every $i$ fulfills $\speclive_i$ while remaining inside $W$, by definition, it fulfills $\speclive_i[W]$, and vice versa.
	Moreover, fulfilling $\safe(W)$ automatically fulfills $\cap_{i=1}^N\specsafe_i$.
	Therefore, the claim follows.
%	 use the decoupled policies obtained from $\G[W]$ (along with the maximally permissive safety policies) as stated above.
%	Finally, by construction, $\speclive_1[W],\ldots,\speclive_N[W]$ are liveness objectives in $\G[W]$, because from all $v\in W$, all objectives are satisfiable.
%	Therefore, we can use the restrictions from the safety objectives to solve the decoupled planning problem in $\G[W]$.
	\qed
\end{proof}

\section{Appendix for~\cref{sec:buchi}}\label{appn: buchi}

\denseUniformBuchi*

\begin{proof}
Let $\G=(V,E,v_0)$, $\sched$, and $\pi_1,\ldots,\pi_N$ be $N$ path-aware policies.
We prove that for each module $i\in[N]$, under the induced measure
$\Pr^\implementation$ on $\pathsinf(\G)$, the event that the objective
$\spec_i$ is satisfied has probability~$1$.
Since a finite intersection of almost-sure events has probability~$1$,
this establishes the theorem.

\smallskip
\noindent
% \textbf{Step~1. (Scheduler fairness and activation blocks.)}
For every history $u\in [N]^\star$
and each index $j\in[N]$, we have $\sched(u)(j)\ge\epsilon>0$.
Fix any module $i$ and let $L_i$ be the uniform bound on its hitting time.
Consider the sequence of scheduler choices $P_0,P_1,\ldots$
defined by the random variables $P_k:[N]^\omega\to[N]$ selecting the
scheduler's output at time~$k$.
For each $n\geq 0$, define the event
\[
  E_n^i \;=\;
  \{\,P_n = P_{n+1} = \cdots = P_{n+L_i-1} = i\,\},
\]
that is,\ the scheduler chooses module~$i$ for $L_i$ consecutive steps starting at position~$n$.
For every history $u\in[N]^n$,
\[
  \Pr\bigl(E_n^i \,\big|\, P_0\ldots P_{n-1}=u\bigr)
  \;\ge\; \epsilon^{L_i}.
\]
Therefore, each event $E_n^i$ has strictly positive conditional probability,
uniformly bounded below by $\epsilon^{L_i}$.

\smallskip
\noindent
\begin{claim}
With probability~$1$, infinitely many of the events
$E_n^i$ occur.    
\end{claim}
\begin{claimproof}
  Due to the events not being independent, instead of using Borel-Cantelli's second lemma directly, we instead define the stopping times
$T_0^i = 0$ and inductively
$T_{k+1}^i = \min\{\,n > T_k^i \mid E_n^i \text{ occurs}\,\}$. 
Because $\epsilon^{L_i}>0$, the conditional
probability that no $E_n^i$ ever occurs after time $T_k^i$ is zero.
This implies that the event
``only finitely many $E_n^i$ occur'' has measure~$0$.
Equivalently,
\[
  \Pr^\implementation(E_n^i \text{ infinitely often}) = 1.
\]
Hence, almost surely, the scheduler produces infinitely many disjoint
blocks of length $L_i$ in which module~$i$ is scheduled exclusively.
\end{claimproof}

\begin{claim}
    Conditioned on $E^i_n$ occuring infinitely often, the B\"uchi condition is satisfied almost-surely.
\end{claim}
\begin{claimproof}
    Fix such a block of length $L_i$ in which only module~$i$ is active.
Assuming uniform bounded hitting time,
starting from any vertex $v$ and any policy $\pi_i$ executed for $L_i$ consecutive steps guarantees that the corresponding path visits the target set $B_i$ with positive probability.
Because there are infinitely many such exclusive activation blocks for
module~$i$, it follows that almost surely the infinite path generated by
the implementation visits $B_i$ infinitely often.
\end{claimproof}

\smallskip
\noindent
For reachability objectives, the same argument applies
(with ``at least once” instead of “infinitely often'').
Finally, since $\spec_1,\ldots,\spec_N$ are finitely many,
\[
  \Pr^\implementation\Bigl(\bigcap_{i=1}^N \spec_i\Bigr)
  = 1 - \Pr^\implementation\Bigl(\bigcup_{i=1}^N \overline{\spec_i}\Bigr)
  = 1,
\]
and hence the implementation $(\G,\pi_1,\ldots,\pi_N,\sched)$
almost surely satisfies the conjunction of all objectives.
\end{proof}

\section{Appendix for~\cref{sec:coBuchi}}\label{appn: coBuchi}
\subsection{Proof of Thm.~\ref{theorem:NoUnivCoBuchi}}
\label{app:NoUnivCoBuchi}

\NoUnivCoBuchi*
\begin{proof}
We give an explicit graph $\G$, specifications $\spec_1$, and $\spec_2$, and policies $\pi_1$ and $\pi_2$, independent of the scheduler $\sigma$, and then show the implementation fails to satisfy the intersection of the specifications.

Let $V = \{v_0,\; v_1,\; v_{11},\; v_{12}\},$
where the initial vertex is $v_0$, and the edges 
$E \:=\:\{\,(v_0,v_0),(v_0,v_1), (v_1,v_{11}),\ (v_1,v_{12}),(v_{11},v_0), (v_{12},v_0)\,\}$.

We define coB\"uchi specifications by their coB\"uchi sets that need to be visited only finitely often where $\spec_j$ requires that the vertex $v_{1j}$ is seen only finitely often along the infinite path.

Note that there are paths that avoid both $v_{12}$ and $v_{21}$ forever (e.g. a path that loops at $v_0$), so $\spec_1\cap\spec_2\neq\emptyset$.

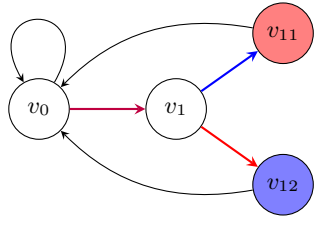
\begin{figure}[h]
\centering
\begin{tikzpicture}[->,>=stealth,node distance=1cm]
  \tikzstyle{normal}=[circle,draw,minimum size=8mm]
  \tikzstyle{bad1}=[circle,draw,fill=red!50,minimum size=8mm]
  \tikzstyle{bad2}=[circle,draw,fill=blue!50,minimum size=8mm]
  \tikzstyle{shared}=[circle,draw,fill=purple!40,minimum size=8mm]

  \node[normal] (v0) {$v_0$};
  \node[normal, right=1cm of v0] (v1) {$v_1$};
  % \node[normal, left=1.5cm of v0] (v2) {$v_2$};

  \node[bad1, right=0.6cm of v1, yshift=+1cm] (v11) {$v_{11}$};
  \node[bad2, right=0.6cm of v1, yshift=-1cm] (v12) {$v_{12}$};

  % \node[shared, left=1cm of v2, yshift=+1cm] (v21) {$v_{21}$};
  % \node[normal, left=1cm of v2, yshift=-1cm] (v22) {$v_{22}$};

  \draw[purple,thick] (v0) -- (v1);
  % \draw (v0) -- (v2);

  \draw[red,thick] (v1) -- (v12);   
  \draw[blue,thick] (v1) -- (v11);  

  % \draw (v2) -- (v21);
  % \draw (v2) -- (v22);

  \draw (v11) to [out=170,in=45] (v0);
  \draw (v12) to [out=190,in=-45] (v0);
\draw (v0) to [out=60,in=120,looseness=8] (v0);

\end{tikzpicture}
\caption{Graph $\G$ used in the proof of \cref{theorem:NoUnivCoBuchi}. Figure is redrawn from \cref{figure:recoBuchiNotUniversal}. Red ($v_{11}$) is the coB\"uchi vertex for $\spec_1$; blue ($v_{12}$) is the coB\"uchi vertex for $\spec_2$; The strategy edge for module $1$ is highlighted in red, and the strategy edge for module 2 is highlighted in blue.}
\label{figure:coBuchiNotUniversal}
\end{figure}

We define deterministic policies $\pi_1,\pi_2:V^\star\to V$ below that satisfy $\spec_1$ and $\spec_2$ respectively.
\[
\begin{array}{ll}
\pi_1(v_0) = v_1, & \pi_1(v_1) = v_{12}, \\
\pi_1(v_{1j}) = v_0\ \text{for }j=1,2
\end{array}
\]
\[
\begin{array}{ll}
\pi_2(v_0) = v_1, & \pi_2(v_1) = v_{11}, \\
\pi_2(v_{1j}) = v_0\ \text{for }j=1,2
\end{array}
\]

%\SS{There is no vertex called \(v_2\) or \(v_{2j}\) for any \(j\) in the figure (or in the description), and the description does not have the edge \((v_2, v_{22})\), so removed them from previous description assuming those are typos/from previous construction}

Consider a scheduler $\sigma$. Since the policies are such that they make the same decision at all vertices other than $v_1$, we focus on the scheduler $\sigma$'s decision at vertex $v_1$.
We deal with two cases. Either one module is eventually always scheduled with non zero probability when the token is at vertex $v_1$, or with probability $1$, both modules are scheduled infinitely often. 

%Either the scheduler schedules scheduler $1$ at $v_1$ infinitely often, that is, the edge $v_{12}$ is taken infinitely often, or $v_2$.

\paragraph{Some module is eventually always scheduled at vertex $v_1$.} for some $j\in\{1,2\}$,
\[
  \Pr^\sigma(\text{eventually only }j)\:=\:\Pr^\sigma\bigl(\exists T\ \forall t\ge T:\ \text{index}_t = j\bigr) > 0.
\]
On the event that the scheduler eventually plays only $j$, the implementation $(\G,\pi_1,\pi_2,\sigma)$ follows the single policy $\pi_j$ from some time on. If \(j=1\) then from that time $T$ onward the execution follows $\pi_1$ and therefore visits $v_{12}$ (the bad state for \(\spec_2\)) infinitely often (indeed on each visit to \(v_1\) it goes to $v_{12}$, so $\spec_2$ is violated on this event. Hence the event ``scheduler eventually plays only 1'' implies violation of $\spec_2$. The probability that both \(\spec_1\) and \(\spec_2\) hold is at most $1-\Pr_\sigma(\text{eventually only }1)<1$. The holds if the scheduler eventually selects only module 2. Thus in either case the intersection is not satisfied almost surely.

\paragraph{No module is eventually always scheduled at vertex $v_1$.}
In this case, $\Pr_\sigma(\exists T\ \forall t\ge T:\ \text{index}_t = j)=0$ for $j=1,2$. Then with probability one the scheduler produces infinitely many occurrences of both indices $1$ and $2$ (otherwise, if some index occurred only finitely often with positive probability that would give a positive-probability lock event). Consequently, along almost all scheduler index-sequences, both modules are scheduled infinitely often.

Under such a scheduler sample path, both indices occur infinitely often, the execution will visit $v_{11}$ infinitely often and $v_{12}$ infinitely often (in fact infinitely many alternating
occurrences typically), so both coB\"uchi conditions are violated.
Therefore in this case the implementation does not satisfy $\spec_1\cap\spec_2$ with probability~$1$.

\end{proof}

\subsection{Proof of Thm.~\ref{thm: coBuchischeduling}}
\label{app:coBuchischeduling}

\coBuchischeduling*
\begin{proof}
	We describe the {\em global} interaction between $N$ agents  as an {\em absorbing} Markov chain \(\L\). Note that all policies have the same set of memory states $M = \pathsfin(\G) \times \mathfrak{C}_B$. A state of the Markov chain is $\zug{\zug{m_1, \ldots, m_N}, v} \in M^N \times V$. 
	%A state is \emph{sink} if \(m_j = m_{j'}\) for all \(j, j' \in \set{1, \ldots, N}\). 
	A state is called \emph{consensus} if $m_j = m_{j'}$, for all $j, j' \in \set{1,\ldots, N}$, dually every other state is reffered to as  \emph{non-consensus}.
	A consensus state \(\zug{\zug{m ,\ldots , m}, v}\) is called a sink if \(m = \zug{\epsilon, C}\) for some cycle \(C\), and the current vertex is already inside \(C\).
	Note that such a sink state corresponds to the case that all policies chooses the same cycle \(C\), thus the path that they generate will be \(C^\omega\), which satisfies all objectives since \(C \in \mathfrak{C}_{B_j}\) for all \(j \in \set{1, 2, \ldots, N}\). 
	Note that such a global state corresponds to the case that all policies choose the same path $\eta$ and cycle $C$, thus the path that they generate will be $\eta \cdot C^\omega$, which satisfies all objectives since $C \in \mathfrak{C}_{B_j}$, for all $j \in \set{1,\ldots, N}$. 
	The probability of transitioning from $\zug{\zug{m_1, \ldots, m_N}, v}$ to $\zug{\zug{m'_1, \ldots, m'_N}, v'}$ is 
	\[\sum_{j=1}^N \sched(j) \cdot \pol_j(v, m_j, \top)( m'_j, v') \cdot \prod_{j' \neq j} \pol_{j'}(v, m_{j'}, \bot)(m'_{j'}, v')\]
	%\SS{Why do we only consider only \(v'\) on the product? The policy which is not scheduled could choose any vertex, isn't it? In particular, shouldn't it be the following: 
		%\[\sum_{j=1}^N \sched(j) \cdot \pol_j(v, m_j, \top)(v', m'_j) \cdot \prod_{j' \neq j} \sum_{v'': E(v, v'')} \pol_{j'}(v, m_{j'}, \bot)(v'', m'_{j'}).\]}
	%We point out that $\pol_j(v, m_j, \top)(v', m'_j)$ is either $0$ or $1$. 

	\begin{claim}
		\(\L\) is indeed a finite Markov chain
	\end{claim}
	
\begin{claimproof}
    We show that the sum over all the outgoing non-zero transition probabilities is indeed \(1\) for any state \(l\) of \(\mathcal{L}\). 
    Fix \(l = \langle \langle m_1, \ldots, m_N \rangle, v \rangle\). 
    Let the next state be \(l' = \langle \langle m'_1, \ldots, m'_N \rangle, v' \rangle\). Then:
    
    \begin{align*}
        \sum_{l'} P(l, l') 
        &= \sum_{\vec{m}' \in M^N, v' \in V} \sum_{j=1}^N \sigma(j) \cdot \pi_j(v, m_j, \top)( m'_j, v') \cdot \prod_{k \neq j} \pi_{k}(v, m_{k}, \bot)(m'_{k}, v') \\
        &= \sum_{j = 1}^{N} \sigma(j) \sum_{v' \in V} \left( \sum_{m'_j \in M} \pi_j(v, m_j, \top)(m'_j, v') \cdot \prod_{k \neq j} \left[ \sum_{m'_k \in M} \pi_{k}(v, m_{k}, \bot)(m'_{k}, v') \right] \right) \\
        \intertext{Note that for \(k \neq j\), the sum over their local next messages \(m'_k\) equals \(1\) regardless of \(v'\). Thus the product term simplifies to \(1\):}
        &= \sum_{j = 1}^{N} \sigma(j) \sum_{v' \in V} \sum_{m'_j \in M} \pi_j(v, m_j, \top)(m'_j, v') \\
        \intertext{Similarly, the inner sum over the active agent's transition leads to 1:}
        &= \sum_{j = 1}^{N} \sigma(j) \cdot 1 \\
        &= 1.
    \end{align*}
\end{claimproof}

We also make the claim below.
	\begin{claim}\label{clm: bscc}
		Every bottom strongly connected components (BSCC) of the Markov chain only consists of sink states, and each sink state is part of some (plausibly singleton) BSCCs.
	\end{claim}
	%we claim that every bottom strongly connected components (BSCC) of \(\G\) only consists of sink states, and each sink state is part of some (plausibly singleton) BSCCs. 
	Since a finite Markov chain eventually reaches a BSCC with probability \(1\), it suffices to conclude that \((\G, \pi_1,\ldots, \pi_N, \sched) \modelsAS \bigwedge_{1 \leq i \leq N} \varphi_i\). \qed

	\begin{claimproof}
		We establish the above claim by showing the following three properties of the constructed Markov chain: (1) the probability of transitioning from any consensus state to a non-consensus one is \(0\) (as a corollary, it holds for the special case of a sink to a non-sink state), (2)
		from any non-consensus state, there exists a path with probability \(> 0\) to some consensus state, 
		% for any non-consensus state, there exists a consensus state such that probability of transitioning from the non-consensus state to the consensus state is \(> 0\), 
		and finally, (3) the probability of transitioning from any non-sink consensus state \(\zug{\zug{m = \zug{\eta, C}, \ldots, m}, v}\) to another non-sink consensus state \(\zug{\zug{m' = \zug{\eta', C}, \ldots, m'}, v'}\) is \(>0\), only if \(\eta' = \eta[1:]\).
		
		(1) implies that a consensus and a non-consensus state cannot be part of same SCC. 
		(2) (together with (1)) establishes that non-consensus states cannot be a part of any BSCC.
		Finally, (3) shows that two non-sink consensus states cannot be in the same SCC.   
		
		%To show (1), 
		%Let us consider an arbitrary consensus state \(\zug{\zug{m, \ldots, m}, v}\), and an arbitrary non-consensus state \(\zug{\zug{m_1', \ldots, m_N'}, v'}\) where \(m_j' \neq m_{j'}'\) for some \(j, j' \in [N]\). 
		%
		%  
		%once the memory states  probability of transitioning  from \(\zug{\zug{m_1, \ldots m_n}, v}\) to \(\zug{\zug{m_1', \ldots, m_n'}, v'}\) is \(0\), when (1) \(m_j = m_{j'}\) for all \(j, j \in [N]\), and (2) there exists \(j, j' \in [N]\) such that \(m_j' \neq m_{j'}'\). 
		
		To show (1), 
		we consider an arbitrary consensus state \(\zug{\zug{m, \ldots, m}, v}\), and an arbitrary non-consensus state \(\zug{\zug{m_1', \ldots, m_N'}, v'}\) where \(m_j' \neq m_{j'}'\) for some \(j, j' \in [N]\). 
		We argue that each summand in the expression of the transition probability is \(0\). 
		Fix an arbitrary \(j \in [N]\). 
		If \(\pi_j(v, m, \top)(m_j', v') = 0\), we are done. 
		Otherwise, we have \(\pi_j(v, m, \top) = ( m_j', v')\) (since \(\pi_j\) is deterministic), and at least one \(j' \in [N]\) such that \(m_j' \neq m_{j'}'\). 
		We consider \(\pi_{j'}(v, m, \bot)(v', m_{j'}')\), and show that this is \(0\).
		Since both \(\pi_j\) and \(\pi_{j'}\) are defined using the same \(\pi\) constructed above, \(\pi_{j'}(v, m, \top) = \pi_j(v, m, \top)\), we have \(\pi_{j'}(v, m, \top) = ( m_j', v')\). 
		Therefore, from \(m_j \neq m_{j'}'\), we have \(\pi_{j'}(v, m, \bot)( m_{j'}', v') = 0\). 
		
		To show (2), we consider an arbitrary non-consensus state \(\zug{\zug{m_1, \ldots, m_N}, v}\), and construct a path with positive probability to a consensus state in the following manner: 
		we first divide the policies \(1, 2, \ldots N\) in at most \(degree(v)\)-many (plausibly empty) sets, namely \(X_1, \ldots, X_{degree(v)}\), where each set \(X_i\) correspond to a neighbouring vertex \(u\) of \(v\) such that all the policies in \(X_i\) proposes \(u\) as its next action.  
		We first assume that at least two of these sets are non-empty, and select a set, namely \(X_l\), with the least number of policies. 
		We then select a policy from \(X_l\), namely policy \(j\). 
		Suppose \(\pi_j(v, m_j, \top) = (v', m')\). 
		We know \(\sigma(j) > 0\). 
		Moreover, for any policy \(j' \in \cup_{i \neq l} X_i\) , \(\pi_{j'}(v, m_{j'}, \bot)(v', m') > 0\).   
		Therefore, there exists an edge with positive probability from \(\zug{\zug{m_1, \ldots, m_N}, v}\) to \(\zug{\zug{m_1', \ldots, m_N'}, v'}\) with \(m_{j'} = m'\) for all \(j' \in \{j\} \cup \bigcup_{l \neq i} X_l\). 
		Since \(X_l\) is of size at most \(\frac{N}{2}\), upon traversing along this edge consensus among at least \(\frac{N}{2}+1\) many policies is achieved. 
		If \(\zug{\zug{m_1', \ldots m_N'}, v'}\) is a consensus state, we are done; otherwise, we apply the same procedure from \(\zug{\zug{m_1', \ldots, m_N'}, v'}\), and so on untill it converges to a consensus state.  
		Note that, in the second step, at least \(\frac{N}{2}+1\) many policies will be in the same set, therefore the size of smallest set could be at most \(\frac{N}{2} -1\). 
		Therefore, the procedure reaches to a consensus state within at most \(N\)-steps, if in every step there are at least two non-empty sets. 
		Finally, we argue that this assumption does not prevent reaching the consensus, albeit delay it. 
		Suppose at some step, there is only a single non-empty set, it means that all the policies agree in their action in that step.  
		If this happens at a non-consensus state, we know eventually their proposed action will differ, so till then we simply traverse along non-consensus states in the Markov Chain without improving the consensusness.  
		%To show (2), we construct a consensus state \(\zug{\zug{m', \ldots, m'}, v'}\) from an arbitrary non-consensus state \(\zug{\zug{m_1, m_2, \ldots, m_n}, v}\) such that the probability of transitioning from the latter to the former is positive. 
		%We fix a \(j \in [N]\), and define \((v', m') \coloneqq \pi_j(v, m_j, \top)\). 
		%For any \(j' \neq j\), denote \(\pi_{j'}(v, m_{j'}, \top) = (v_{j'}'', m_{j'}'')\). 
		%If \(v_{j'}'' = v'\), then \(\pi_{j'}(v, m_{j'}, \bot)(v', m') = 1\), otherwise 

		Finally, showing (3) is straightforward: since, every policy has the same memory state at a given vertex \(v\), no matter which policy is scheduled, each of them updates along the path to the cycle. 
		Since by assumption, we consider only simple path to the cycle, there can be edge with probability \(> 0\) in only one of the direction. 
	\end{claimproof}
\end{proof}

\section{Appendix for~\cref{sec:parity}}\label{appn: parity}
\subsection{Proof of Lem.~\ref{lemma: kappagood}}
\label{app:kappagood}
	\begin{proof}
		Consider an infinite path \(\gamma\) that is good for all pairty objective \(\kappa_i\). 
		Let \(a_i \in \inf(\gamma)\) for each \(i \in [N]\) be the states that maximises \(\kappa_i\). 
		Since \(\gamma\) satisfies each \(\kappa_i\)'s, \(\kappa_i(a_i)\) is even. 
		Since each of \(a_i\) is visited infinitely often in \(\gamma\), without loss of generality, we assume that there is a path \(\tau_i\) from \(a_i\) to \(a_{i+1}\) for \(i = [N-1]\), and a path \(\tau_N\) from \(a_N\) to \(a_1\) such that each vertex inside \(\tau_i\) is visted infinitely often in \(\gamma\) for every \(i \in [N]\). 
		Obtain \(\tau_i'\) from \(\tau_i\) by deleting all 
		the cycles.
		Observe the cycle \(\tau_1' \cdot \tau_2' \cdots \tau_N'\). 
		The maximal parity index in \(\kappa_i\) that the cycle traverses is \(\kappa_i(a_i)\) for each \(i\). 
		In other words, \((\tau_1' \cdot \tau_2' \cdots \tau_N')^\omega\) satisfies each parity objective \(\kappa_i\).  
		
		Fix an arbitray \(i \in [N]\), and we describe \(\theta_i\): for any vertex \(a \in \tau_i'\), and for its successor vertex \(a' \in \tau_i'\), \(\theta_i(a) = a'\); for any state \(a \in \gamma_{(i+1) \mod N}\) for which \(\theta_i\) is yet to be defined, and its successor \(a' \in \gamma_{i+1\mod N}\),  we define \(\theta_i(a) = a'\), and so on till we exhaust all the vertices of \(\gamma_{i+(N-1) \mod N}\). 
		For any other vertex \(a \in V\), we define all \(\theta_i\) to choose the same successor, one that leads closer to the cycle \(\tau_1' \cdot \tau_2' \cdots \tau_N'\).
		Note that, each \(\theta_i\) is memoryless.  
		
		By asumption \(\sigma\) schedules every agent infinitely often at every vertex, thus we note that in any path \(\delta\) with non-zero measure, \(\inf(\delta)\) is the set of vertices of the cycle \(\tau_1' \cdots \tau_N'\). 
		Since, the cycle satisfies all the parity objectives \(\kappa_i\), the implementation almost-surely satisfies \(\kappa_i\) for each \(i \in [N]\). \qed
	\end{proof}

\subsection{Proof of Thm.~\ref{thm: parityscheduling}}
\label{app:parityscheduling}

\parityscheduling*

\begin{proof}
	We describe the \emph{global} interaction between \(N\) policies as an absorbing Markov chain. 
	Note that all policies have the same set of memory states \(M = (V^V)^N\). 
	
	A state of the Markov chain is \(\zug{\zug{m_1, m_2, \ldots m_N}, v} \in (V^V)^N \times V\), where each \(m_i = \zug{\thetatwo{i}{1}, \thetatwo{i}{2}, \ldots \thetatwo{i}{N}}\). Recall that, \(\thetatwo{i}{j}\) denotes the memoryless policy that Agent \(i\) thinks Agent \(j\) would play (for \(i = j\), it is the policy Agent \(i\) intends to play if scheduled).
	
	For each state \(l\) of \(\L\), we define a set \emph{used-edges} \(U_l\) as follows: an edge \(e = (u, w)\) is in \(U_l\) if there is an agent \(i\) such that \(\thetatwo{i}{i}(u) = v\). 
	A vertex \(u\) is called \emph{conflicting} at a state \(l\) if there exists Agent \(i, j\) such that \(\thetatwo{i}{i}(u) \neq \thetatwo{i}{j}(u)\). 
	A state \(l\) is called consensus if in the graph \(\G\) restricted to the set used-edges \(U_l\) of \(l\), a conflicting vertex \(u\) of \(l\) is not reachable from \(v\). 
	Intuitively speaking, a state is called consensus if every agent correctly guesses about every other agents's memoryless policies, possibly except for the vertices which are anyway not reachable from \(v\).  
	Every other states are referred to as non-consensus. 
	Once it reaches a consensus state, none of the memory states ever change. 
	Recall that implementation of a fair scheduler on each memory state \(m_i\), by design,  almost surely satisfies \(\kappa_i\).
	Therefore, when \(m_i = m\) for some memory state \(m\) for all \(i \in [N]\), the implementation of \(\sigma\) on \(\zug{m, \ldots, m}\) almost surely satisfies all the parity objectives. 
	It remains to show that no matter from which state of the Markov chain the procedure starts, it always reaches a consensus state. 
	We first define the transition probabilities from \(\zug{\zug{m_1, \ldots , m_N}, v}\) to \(\zug{\zug{m_1', \ldots m_N'}, v'}\) as follows:
	\[
	\sum_{j = 1}^{N} \sigma(j) \cdot \pi_j( v, m_j, (j, v'))(m_j', v') \prod_{j' \neq j} \pi_{j'}( v, m_{j'}, (j, v'))(m_{j'}', v') 
	\]
	We first show that it is indeed a Markov Chain. 
	
	\begin{claim}
		\(\L\) is a Markov chain. 
	\end{claim}

	\begin{claimproof}
		We show that the sum over all the outgoing non-zero transition probabilities is indeed \(1\) for any state \(l\) of \(\L\). 
		Fix an \(l = \zug{\zug{m_1, \ldots m_N}, v}\). 
		Then, 
		
		\begin{align*}
			\sum_{\zug{\zug{m_1', \ldots, m_N'}, v'} \in M^N \times V}  &\sum_{j=1}^N \sched(j) \cdot \pol_j(v, m_j, (j, v'))( m'_j, v') \cdot \prod_{j' \neq j} \pol_{j'}(v, m_{j'}, (j, v'))(m'_{j'}, v')\\
			&= \sum_{j = 1}^{N} \sigma(j) \sum_{ M^N \times V} \pi_j(v, m_j, (j, v'))(m_j, v') \cdot \prod_{j' \neq j} \pol_{j'}(v, m_{j'}, (j, v'))(m'_{j'}, v')\\
			&= \sum_{j = 1}^{N} \sigma(j) \sum_{ M^{N-1} \times V}  \cdot \prod_{j' \neq j} \pol_{j'}(v, m_{j'}, (j, v'))(m'_{j'}, v')\\
			&= \sum_{j = 1}^{N} \sigma(j) \sum_{ M^{N-2} \times V} \prod_{j'' \neq \{j, j'\}} \pi_{j''}(v, m_{j''}, (j, v'))(m_{j''}', v')\sum_{M} \pi_{j'}(v, m_{j'}, (j, v'))(m_{j'}', v')\\
			&=\sum_{j = 1}^{N} \sigma(j) \sum_{ M^{N-2} \times V} \prod_{j'' \neq \{j, j'\}} \pi_{j''}(v, m_{j''}, (j, v'))(m_{j''}', v')\\
			&\qquad \vdots\\
			&= \sum_{j = 1}^{N} \sigma(j) = 1
		\end{align*}
	\end{claimproof} 
	
	Since a finite Markov chain eventually reaches a BSCC with probability \(1\), it suffices to show that consesus states are BSCCs of this Markov Chain. 
	Formally, 
	\begin{claim}
		Every BSCC of \(\L\) only consists of consensus states.
	\end{claim}
	
	\begin{claimproof}
		We establish the above statement by showing the following three properties: (1) The transition probability from any consensus state to any non-consensus state is \(0\),  (2) For any non-consensus state, there exists a path with probability \(> 0\) to some consensus state, and finally, (3) there exists at least one consensus state in \(\L\). 
		Since, \(\L\) has no dead-end by design, this proves the claim. 
		
		To show (1), 
		we consider an arbitrary consensus state \(\zug{\zug{m, \ldots, m}, v}\), and an arbitrary non-consensus state \(\zug{\zug{m_1', \ldots, m_N'}, v'}\) where \(m_j' \neq m_{j'}'\) for some \(j, j' \in [N]\). 
		We argue that each summand in the expression of the transition probability is \(0\). 
		Fix an arbitrary \(j \in [N]\). 
		If \(\pi_j(v, m, (j, v'))(m_j', v') = 0\), we are done. 
		Otherwise, we have \(\pi_j(v, m, (j, v')) = (m_j', v')\) (since \(\pi_j\) is deterministic in this case), and at least one \(j' \in [N]\) such that \(m_j' \neq m_{j'}'\) (since we consider a non-consensus state). 
		We consider \(\pi_{j'}(v, m, (j, v'))(m_{j'}', v')\), and show that this is \(0\).
		Since both \(\pi_j\) and \(\pi_{j'}\) are defined using the same \(\pi\) constructed above, \(\pi_{j'}(v, m, (j, v')) = \pi_j(v, m, (j, v'))\), we have \(\pi_{j'}(v, m, (j, v')) = (m_j', v')\). 
		Therefore, from \(m_j \neq m_{j'}'\), we have \(\pi_{j'}(v, m, (j, v'))(v', m_{j'}') = 0\).
		
		To show (2), we consider an arbitrary non-consensus state \(l = \zug{\zug{m_1, \ldots, m_N}, v}\), with \(m_i = \zug{\thetatwo{i}{1}, \ldots \thetatwo{i}{N}}\) and construct a path with positive probability to a consensus state in the following. 
		At any state \(s\), for each agent \(i\), we denote by \(X_i^{(s)}\) the set of agents \(i'\) for which \(\thetatwo{i}{i} \neq \thetatwo{i'}{i}\), and by \(Y_i^{(s)}\) the set of agents \(i\) for which \(\thetatwo{i}{i}(u) \neq \thetatwo{i'}{i}(u)\), where \(u\) is the location at \(s\).
		Clearly, \(Y_i^{(s)} \subseteq X_i^{(s)}\) for any \(i\) and \(s\).  
		Since \(l\) is a non-consensus state, there exists at least one agent \(i\) for which \(X_i^{(l)}\) is non-empty.  
		Otherwise, we are already at a consensus state!
		We know \(\sigma(i) > 0\). 
		Without loss of generality, we assume that \(Y_i^{(l)}\) is also non-empty.
		We consider the scenario when Agent \(i\) is scheduled, which makes the next vertex \(\thetatwo{i}{i}(v) = v'\). 
		Since every agent \(i'\) of \(Y_i^{(l)}\) makes a wrong guess about Agent \(i\), with positive probability, all of them changes their memory state to \(m_i\). 
		In other words, there exists an edge from \(l\) to \(l' = \zug{\zug{m_1', \ldots, m_N'}, v'}\), where \(m_{i'}' = m_i\) for all \(i' \in Y_i^{(l)}\), making \(X_i^{(l')} = X_i^{(l)} \setminus Y_i^{(l)}\). 
		If \(l'\) is a consensus state, we are done. 
		Otherwise, we select another agent \(j\) for which \(X_j^{(l')}\) is non-empty, and continue the procedure.
		 
		Note that,  
		once two agents reaches a consensus in their memory state, with positive probability, they never lose that consensus again. 
		This is because from that point onwords, either both of them requires change of memory (when both of their guessed policy about some third agent turns out to be wrong) or none of them changes policy.  
		So, when they require to change, they can change their memory states to yet another same memory state with positive probability.
		 
		Therefore, in this path, for any state \(s\), the size of at least one \(X_i^{(s)}\) is decreasing if there is at least one \(i\) with non-empty \(Y_i^{(s)}\).	  
		If \(Y_i^{(s)}\) is empty for all \(i\) at some non-consensus state\(s\) at some step, with positive probability, we add an edge to the path where the consensus remains the same. 
		Since by definition of consensus state, a conflicting vertex \(u\) is always reachable, we follow that path from \(v\) by scheduling the agents appropriately. 
		Yet again, we can do this with probability \(> 0\). 
		In this manner, we have a path with positive probability from a non-consensus state \(l'\) with no conflicting vertex to a non-consensus vertex \(l''\) with conflicting vertices.  
		From \(l''\), we follow the previously mentioned procedure to go to a state with improved consensus with positive probability.
		Iterating these two procedures, we get a path from non-consensus state to a consensus state with probability \(> 0\). 
		
		Finally, property (3) is established in Lemma~\ref{lemma: kappagood}. 
	\end{claimproof}
	
\end{proof}

\end{document}